\documentclass[twocolumn,letter]{IEEEtran}
\usepackage{latexsym}
\usepackage{amssymb}
\usepackage{amsmath}

\usepackage{bm}
\usepackage[dvips]{graphicx}

\newtheorem{thm}    {Theorem}%[section]
\newtheorem{lem}     {Lemma}%[section]
\newtheorem{proposition}        {Proposition}%[section]
\newtheorem{rem}     {Remark}%[section]
\newtheorem{example}  {Example}%[section]
\newtheorem{condition}  {Condition}%[section]

\def\argmax{\mathop{\rm argmax}}
\def\mix{\mathop{\rm mix}}

\newcommand{\defeq}{\stackrel{\rm def}{=}}

\newcommand{\bF}{\mathbb{F}}

\def\cA{{\cal A}}
\def\cE{{\cal E}}
\def\cY{{\cal Y}}
\def\cZ{{\cal Z}}

\def\cM{{\cal M}}
\def\cD{{\cal D}}

\def\rE{{\rm E}}
\def\rP{{\rm P}}

\newcommand{\cX}{{\cal X}}

\newcommand{\bX}{{\bf X}}
\newcommand{\bY}{{\bf Y}}

\makeatletter
\newcommand{\lleq}{\mathrel{\mathpalette\gl@align<}}
\newcommand{\ggeq}{\mathrel{\mathpalette\gl@align>}}
\newcommand{\gl@align}[2]{%\lower.1ex
\vbox{\baselineskip\z@skip\lineskip\z@
\ialign{$\m@th#1\hfil##\hfil$\crcr#2\crcr{}_{{}_{(=)}}\crcr}}}
\makeatother

\def\Label#1{\label{#1}\ \text{[\ #1\ ]}\ }
\def\Label{\label}

\begin{document}
\title{Tight exponential analysis of universally composable privacy amplification and its applications}
\author{
Masahito~Hayashi~\IEEEmembership{Senior Member,~IEEE}
\thanks{
M. Hayashi was with Graduate School of Information Sciences, Tohoku University, Aoba-ku, Sendai, 980-8579, Japan
and
Centre for Quantum Technologies, National University of Singapore, 3 Science Drive 2, Singapore 117542.
He is now with 
Graduate School of Mathematics, Nagoya University,
Furocho, Chikusaku, Nagoya, 464-860, Japan
and
Centre for Quantum Technologies, National University of Singapore, 3 Science Drive 2, Singapore 117542.
(e-mail: masahito@math.nagoya-u.ac.jp)}}
\date{}
\maketitle

\begin{abstract}
Motivated by the desirability of universal composability, we analyze in terms of $L_1$ distinguishability 
the task of secret key generation from a joint random variable. 
Under this secrecy criterion, 
using the R\'{e}nyi entropy of order $1+s$ for $s \in [0,1]$,
we derive a new upper bound of
Eve's distinguishability under the application of the universal$_2$ hash functions.
It is also shown that this bound gives the tight exponential rate of decrease in 
the case of independent and identical distributions.
The result is applied to the wire-tap channel model 
and to secret key generation (distillation) by public discussion.
\end{abstract}

\section{Introduction}
Random privacy amplification based on the universal$_2$ condition \cite{Carter}
has been studied by many authors \cite{BBCM,MW,RW,Ren05,ILL,Hayashi2}.
This technique is originally developed for random number extraction \cite{BBCM,MW}. 
It can also be applied to 
secret key generation (distillation) with public communication \cite{AC93,Mau93,MUW05,MM,WMU,MW,RW} 
and the wire-tap channel \cite{Wyner,CK79,Csiszar,Deve,WNI,Hayashi}, which treats 
the secure communication in the presence of an eavesdropper.
(For details of its application, see e.g. the previous paper \cite{Hayashi2}.)
When random privacy amplification 
is implemented with universal$_2$ hash functions,
it can yield protocols for the above tasks
with a relatively small amount of calculation.
%This method is effective for the discrete memoryless case.

Similar to the study \cite{BBCM,ILL} for random privacy amplification based on the universal$_2$ condition,
the previous paper \cite{Hayashi2} focused only on the mutual information with the eavesdropper.
However, 
as the secrecy criterion,
many papers in the cryptography community
 \cite{Can,MW,RW,Ren05} adopt the half of the $L_1$ norm distance, so called $L_1$ distinguishability
because this criterion is closely related to universally composable security \cite{Can}.
In this paper, we adopt $L_1$ distinguishability as the secrecy criterion, and evaluate the secrecy for random privacy amplification.
%while the previous paper \cite{Hayashi2} adopts the mutual information criterion.
In the independent and identically distributed case,
when the rate of generated random numbers is smaller than the entropy of the original information source,
it is possible to generate a random variable whose $L_1$ norm distance to the uniform random number approaches zero asymptotically.
In the realistic setting, we can manipulate only a finite size of random variables.
In order to treat the performance in the finite length setting,
we have two kinds of formalism for the independent and identical distribution setting.

The first one is the second order formalism, in which,
we focus on the asymptotic expansion up to second order in $\sqrt{n}$
of the length of the generated keys $l_{n}$ 
as $l_{n}= H n + C\sqrt{n}+o(\sqrt{n})$
with a constant constraint for the security parameter.
The second one is the exponent formalism,
in which, we fix the generation rate $R:= l_{n}/n$
and evaluate the exponential decreasing rate of convergence of the security parameter.
In the exponent formalism,
it is not sufficient to show that the security parameter goes zero exponentially,
and it is required to explicitly give lower and/or upper bounds for the exponential decreasing rate.
The exponent formalism has been studied by various information theoretical settings,
e.g., channel coding \cite{Gal,SGB}, source coding \cite{CKbook,Han-source}, and mutual information criterion in wire-tap channel \cite{Hayashi,Hayashi2}.
As for the second order formalism, 
the optimal coding length with the fixed error probability
has been derived up to the second order $\sqrt{n}$ in various settings \cite{strassen,Hay1,Pol}
in the case of channel coding.
In particular, the previous paper \cite{Hay1} treats it based on the information spectrum approach \cite{Han}, 
which is closely related to $\epsilon$-smooth min-entropy.
Note that, as is mentioned by Han \cite{Han},
the information spectrum approach cannot yield the optimal exponent of error probability in the channel coding.

Concerning the second order formalism for uniform random number generation,
the previous paper \cite{H08} has solved the optimal second order coefficient
under $L_1$ distinguishability criterion and other criteria
by employing the information spectrum method when there is no side information. 
Even when the side information exists, 
the same argument can be shown for the second order formalism
by replacing 
the variance of the likelihood by 
the variance of the likelihood for the conditional distribution
due to the following reason.
For the converse part,
the key lemma (\cite[Lemma 4]{H08},\cite[Lemma 2.1.2]{Han}) holds
by replacing the distribution by the conditional distribution.
The direct part
can be shown by replacing 
the key lemma (\cite[Lemma 3]{H08},\cite[Lemma 2.1.1]{Han}) 
by the inequality (\ref{speLem}) in the present paper, which holds under the universal$_2$ hash functions.

However, the exponent formalism for $L_1$ distinguishability with secure key generation
has not been studied sufficiently.
Only the previous paper \cite{Hayashi2} treated it with mutual information criterion.
Therefore, the present paper focuses on the exponent formalism for $L_1$ distinguishability.

%In the community of information theory, 
%in order to discuss the speed, we often focus on the the exponential rate of decrease.
%This rate is called the exponent, and is widely discussed among several topics in information theory,
%e.g., channel coding \cite{Gal}, source coding \cite{CK79,Han-source}, and mutual information criterion in wire-tap channel \cite{Hayashi,Hayashi2}.
%However, the exponent has not been discussed in the community of cryptography as an important criterion.
%The purpose of this paper is establishing a systematic evaluating method for exponent for the $L_1$ norm distance in secure protocols.

In Section \ref{s4},
first,
we focus on evaluation for random privacy amplification by Bennett et al \cite{BBCM},
which employs the R\'{e}nyi entropy of order 2.
This evaluation was also obtained by H\r{a}stad et al \cite{ILL} and is often called leftover hash lemma.
Using a discussion similar to Renner \cite{Ren05},
we derive an upper bound for the $L_1$ norm distance
under the universal$_2$ condition for hash functions, 
which is the main theorem of this paper (Theorem \ref{thm1}).

Next, we apply this theorem to the i.i.d. setting with a given key generation rate and a given source distribution.
Then, 
we derive a lower bound of the exponent of the average of the $L_1$ norm distance between the generated random number and the uniform random number
when a family of universal$_2$ hash functions is applied.
Next, we introduce a stronger condition for hash functions, which is called strongly universal$_2$.
We consider the $n$-independent and identical extension,
and show that the exponential rate of decrease for this bound is tight under a stronger condition
by using the type method,
which was invented by Csisz\'{a}r and K\"{o}rner \cite{CKbook} and is one of standard methods in information theory.
Since our bound realizes the optimal exponent, 
it gives a powerful bound even for the finite length setting \cite{Watanabe}.
One might consider that the smooth min entropy can derive the same lower bound for the exponential decreasing rate of universal composability.
However, as shown in Subsection \ref{s2d}, 
the bound derived by the smooth min entropy is strictly smaller than that by smoothing of R\'{e}nyi entropy of order $2$.
This disagreement is not so unnatural because 
a similar disagreement appears for the exponent of error probability in the channel coding
as a relation between Gallager exponent and the lower bound derived by the information spectrum approach \cite{Han}.

Further, if our protocol generating the random number is allowed to depend on the original distribution,
there is a possibility to improve the exponent while it is known that asymptotic generation cannot be improved \cite{VV}.
In Section \ref{s3}, we derive the optimal exponent in this setting by using Cram\'{e}r's Theorem \cite{DZ} and the type method \cite{CKbook}.
Based on comparison between this exponent and the exponent given in Section \ref{s4},
we can compare the performances between 
the protocol taking into account the full probability distribution of the source
and the protocol based on the entropy of the source, which is realized by 
universal$_2$ hash functions.

In Section \ref{s42}, we consider the case when an eavesdropper has a random variable correlated to the random variable of the authorized user.
In this case, applying universal$_2$ hash functions to his random variable,
the authorized user obtain a secure random variable.
We apply our evaluation of $L_1$ norm distance obtained in Subsection \ref{s2a} (Theorem \ref{thm1})
to the distribution of the authorized user 
when the eavesdropper's random variable is fixed to a certain value.
Then,
we obtain a tighter evaluation (\ref{4-25-1})
than that directly obtained from the previous paper \cite{Hayashi2}.

In Section \ref{s2}, 
we focus on wire-tap channel model, whose capacity has been calculated by
Wyner \cite{Wyner} and Csisz\'{a}r and K\"{o}rner \cite{CK79}.
Csisz\'{a}r \cite{Csiszar} showed the strong security,
and many papers \cite{Hayashi2,BL,CTD} treat this model with mutual information criterion.
The previous paper \cite{Hayashi} derived bounds for both exponential rates of decrease 
for the security criterion based on the $L_1$ norm distance
as well as the mutual information between Alice and Eve.
It obtained a bound for the exponential rate of decrease concerning the $L_1$ security criterion.
In this paper, we apply (\ref{4-25-1}) to wire-tap channel model,
and obtain the evaluation of the exponent of the $L_1$ security criterion.
In Section \ref{s6-1}, 
it is shown 
that the evaluation obtained in this paper is better than that by the previous paper \cite{Hayashi}.
In a realistic setting, it is natural to restrict our codes to linear codes.
In Section \ref{s3-a}, 
using (\ref{8-5-1}), we provide a security analysis for a code constructed by the combination of 
an arbitrary linear code and privacy amplification by universal$_2$ hash functions.
This analysis yields the exponential rate of decrease for the $L_1$ security criterion.
Overall, since (\ref{4-25-1}) and (\ref{8-5-1}) are derived from Theorem \ref{thm1},
all of the obtained results concerning the wire-tap channel model can be regarded as consequences of Theorem \ref{thm1}.

Further, 
in Section \ref{s6}, 
we obtain the bound for the $L_1$ security criterion
in one-way secret key generation.
In Appendix \ref{as2}, we prove Theorem \ref{thm2} mentioned in Subsection \ref{s2a}.
In Appendix \ref{al9-9-1}, we prove Lemma \ref{l9-9-1} given in Subsection \ref{s3}.
In Appendix \ref{8-03-e}, we show Equation (\ref{8-03-a}), which is important for comparison in Subsection \ref{s2d}.

%{\it Relation with the previous paper \cite{Hayashi2}:}~
\subsection*{Relation with the previous paper \cite{Hayashi2}}
The main difference from the previous paper \cite{Hayashi2} is that 
the analysis on this paper is based on $L_1$ distinguishability while that on the previous paper \cite{Hayashi2}
is based on the mutual information criterion.
In the first step, this paper derives an evaluation (Theorem \ref{thm1}) of the equality of the uniform random number generation 
by universal$_2$ hash functions based on the $L_1$ norm criterion.
Applying Theorem \ref{thm1}, we treat several security problems.
Since this paper treats the same security problems as the previous paper with the different criterion,
some of protocols used in this paper were used in the previous paper \cite{Hayashi2}.
That is, the coding protocols used in Sections \ref{s2}, \ref{s3-a} and \ref{s6}
are used in Sections III, V, and VI in \cite{Hayashi2}, respectively.
While these protocols are described in \cite{Hayashi2}, 
we describe the whole protocols in this paper for the readers' convenience.

For uniform random number generation, 
this paper gives the tight exponential rate of decrease for the $L_1$ norm distance,
while the previous paper \cite{Hayashi2} gives a lower bound on the exponential rate of decrease based on Shannon entropy.
Concerning secret key generation without communication,
this paper gives a lower bound of the exponential rate of decrease based on $L_1$ distinguishability,
while the previous paper \cite{Hayashi2} gives a lower bound of the exponential rate of decrease based on the mutual information criterion.
Applying Pinsker's inequality (\ref{1-5-3-1}), we can derive 
a lower bound of the exponential rate of decrease based on $L_1$ distinguishability 
from the lower bound in \cite{Hayashi2}.
As is shown in Lemma \ref{L1-6-2} in Subsection \ref{s42s},
our lower bound is (strictly) better than combination of Pinsker's inequality and the lower bound by \cite{Hayashi2}
(except for special cases).
Note that 
application of 
Pinsker's inequality (\ref{1-5-3-1}) or (\ref{1-5-3})
yields the half of the lower bound of the exponent of the mutual information as a lower bound of the exponent of 
universal composability.
Indeed, we give a numerical example in Fig. \ref{f3}, in which
our bound is strictly better than that by \cite{Hayashi2}.

Concerning the wire-tap channel in a general framework,
the code given in this paper is quite similar to that in the previous paper \cite{Hayashi2}.
However, the evaluation method in this paper is 
different from that of the previous paper \cite{Hayashi2}
because the analysis in this paper is based on $L_1$ distinguishability 
while that in the previous paper \cite{Hayashi2} is based on the mutual information.
In this model, 
we can derive a lower bound for the exponential rate of decrease based on $L_1$ distinguishability
by the combination of Pinsker's inequality (\ref{1-5-3-1})
and the result in \cite{Hayashi2}.
As is shown in Section \ref{s6-1},
our lower bound is 
better than this lower bound from \cite{Hayashi2}.
Section \ref{s3-a} treats a more realistic setting
by using linear codes.
Even in this setting, 
as is explained in Remark \ref{rem1},
our lower bound is strictly better than the lower bound by \cite{Hayashi2} (except for special cases mentioned in Lemma \ref{L1-6-2}).
The same observation can be applied to 
secret key generation by public communication,
which is discussed in Section \ref{s6}.

\section{Preliminaries}
First, we briefly explain some notation and basic knowledge in 
information theory.
In order to evaluate the difference between two distributions
$P^X$ and $\tilde{P}^X$,
we employ the following quantities:
the $L_1$ distance (variational distance)
\begin{align}
 d_1(P^X,\tilde{P}^X) 
:=
\sum_{x}|P^X(x)-\tilde{P}^X(x)|,
\end{align}
the $L_2$ distance
\begin{align}
d_2(P^X,\tilde{P}^X) 
:=\sqrt{\sum_{x}(P^X(x)-\tilde{P}^X(x))^2},
\end{align}
and the KL-divergence
\begin{align}
D(P^X \| \tilde{P}^X) 
:=\sum_{x} P^X(x)(\log P^X(x)-\log \tilde{P}^X(x)),
\end{align}
where $\log$ expresses the natural logarithm.
These definitions can be extended 
when the total measure is less than $1$ i.e.,
$\sum_{a} P^{A}(a) \le 1$.
In the following, we call such $P^A$ a sub-distribution.
This extension for sub-distributions
is crucial for the later discussion.

When a joint distribution $P^{X,Y}$ is given,
we have the following equation
\begin{align}
& d_1(P^{X,Y},\tilde{P}^X \times P^Y) 
=
\sum_{x,y} |P^{X,Y}(x,y)- \tilde{P}^X (x) P^Y (y)| \nonumber \\
=&
\sum_{y} P^Y(y) \sum_x |P^{X|Y}(x|y)- \tilde{P}^X (x) | \nonumber \\
= &
\sum_{y} P^Y(y) d_1 (P^{X|Y=y},\tilde{P}^X).\Label{1-4-1}
\end{align}

When $P^X,\tilde{P}^X$ are normalized distributions,
as a relation between
the KL-divergence
and the $L_1$ distance,
the Pinsker's inequality 
\begin{align}
\frac{1}{2} d_1(P^X,\tilde{P}^X) ^2 \le
D(P^X \| \tilde{P}^X) 
\Label{1-5-3-1}
\end{align}
is known \cite{CKbook}.
That is,
\begin{align}
-\log d_1(P^X,\tilde{P}^X) 
\ge
\frac{-1}{2}(\log D(P^X \| \tilde{P}^X) 
+\log 2).
\Label{1-5-3}
\end{align}
These relations will be helpful for later discussions.

\section{Uniform random number generation}\Label{s4}
\subsection{Protocol based on universal$_2$ hash function: Direct part}
\Label{s2a}
%\subsection{Method based on R\'{e}nyi entropy of order $1+s$}
Firstly,
we consider the uniform random number generation problem from
a biased random number $a \in \cA$, which obeys a probability 
distribution $P^A$
for finite cardinality $|\cA|$.
There are two types of protocols for this problem.
One is a protocol specialized for the given distribution $P^A$.
The other is a universal protocol that does not depend on 
the given distribution $P^A$.
%This type protocol depends only on the generation rate $R$.
The aim of this section is evaluate the performance of the latter setting.
In the latter setting, our protocol is given by a function $f$ from $\cA$ to $\cM=\{1,\ldots, M\}$.

The quality of the random number obeying the sub-distribution $P^A$ is 
evaluated by 
\begin{align}
d_1(P^{A}):=
d_1(P^{A}, P^{A}({\cal A}) P^{A}_{\mix}) ,
\end{align}
where $P^{A}_{\mix}$ is the uniform distribution on $\cA$.
We also use 
the R\'{e}nyi entropy of order $1+s$:
\begin{align*}
H_{1+s}(A|P^{A}) 
:=\frac{-1}{s}\log \sum_{a} P^{A}(a)^{1+s}.
\end{align*}
The $L_2$ distance is written by using the R\'{e}nyi entropy of order $2$ as follows.
\begin{align}
d_2(P^{A}, P^{A}({\cal A}) P^{A}_{\mix}) ^2
=e^{-H_{2}(A|P^{A})}- \frac{P^{A}(\cA)^2}{|\cA|}.
\Label{5-14-3}
\end{align}

Now, we focus on an ensemble of functions $f_{\bX}$ from 
$\cA$ to $\cM=\{1, \ldots, M\}$, 
where $\bX$ denotes a random variable describing 
the stochastic behavior of the function $f_{\bX}$.
In this case, we adopt on the following quantity
as a criterion of the secrecy:
\begin{align}
&\rE_{\bX} d_1(P^{f_{\bf X}(A)}) 
=
\rE_{\bX}
d_1(P^{f_{\bf X}(A)}, P^{A}({\cal A}) P^{f_{\bf X}(A)}_{\mix}) \nonumber \\
=& 
d_1(P^{B,\bX}, P^{A}({\cal A}) P^{B}_{\mix}\times P^{\bX}) ,
\Label{1-5-4}
\end{align}
where $B$ is the random variable $f_{\bf X}(A)$
and the final equation follows from (\ref{1-4-1}).
Hence, when the expectation $\rE_{\bX} d_1(P^{f_{\bf X}(A)})$
is sufficiently small,
the random variable $f_{\bf X}(A)$ is almost independent of 
the side information $\bX$.
Then, 
the choice $f_{\bf X}$ can be communicated between Alice and Bob 
without revealing anything about $f(A)$.

An ensemble of hash functions $f_{\bX}$ is called universal$_2$ when it satisfies the following condition \cite{Carter}:
\begin{condition}[Universal$_2$]
\Label{C1}
For any elements $a_1 \neq  a_2\in \cA$,
the collision probability that $f_{\bX}(a_1)=f_{\bX}(a_2)$ is
at most $\frac{1}{M}$.
\end{condition}
We sometimes require 
the following additional condition:
\begin{condition}\Label{C12}
For any $\bX$, the cardinality of $f_{\bX}^{-1}\{i\}$ does not depend on $i$.
\end{condition}
This condition will be used in Section \ref{s3}.

Indeed, when the cardinality $|\cA|$ is a power of a prime power $q$
and $M$ is another power of the same prime power $q$,
as is shown in Appendix II of the previous paper \cite{Hayashi2},
the ensemble $\{f_{\bX}\}$ can be chosen to be
the concatenation of a Toeplitz matrix and the identity 
$(\bX,I)$ \cite{Krawczyk}
only with $\log_q |\cA|-1$ random variables taking values in the finite field $\bF_q$.
That is, the function can be obtained by the multiplication of the random matrix $(\bX,I)$ taking values in $\bF_q$.
In this case, 
Condition \ref{C12} can be confirmed because the rank of $(\bX,I)$ is constant.
%For Condition \ref{C1}, see Appendix II of the previous paper \cite{Hayashi2}.

%However, we can consider Conditions \ref{C1} and \ref{C12} even if the cardinality $|\cA|$ is infinite.

Bennett et al \cite{BBCM} essentially showed the following lemma.
\begin{lem}\Label{LL1}
A family of universal$_2$ hash functions $f_{\bX}$ satisfies
\begin{align}
\rE_{\bX}
e^{-H_{2}(f_{\bf X}(A)|P^{f_{\bf X}(A)})}
\le 
e^{-H_{2}(A|P^{A})}+
\frac{P^{A}(\cA)^2}{M}.\Label{5-14-4}
\end{align}
\end{lem}
This was also shown by H\r{a}stad et al \cite{ILL} and is often called leftover hash lemma.

Now, we follow the derivation of Theorem 5.5.1 of Renner \cite{Ren05}
when one classical random variable is given.
The Schwarz inequality implies that
\begin{align*}
& d_1(P^{f_{\bf X}(A)}, P^{A}({\cal A}) P^{f_{\bf X}(A)}_{\mix}) \\
\le &
\sqrt{M}
\sqrt{
d_2(P^{f_{\bf X}(A)}, P^{A}({\cal A}) P^{f_{\bf X}(A)}_{\mix}) 
}.
\end{align*}
Jensen's inequality yields that
\begin{align*}
& \rE_{\bX}
d_1(P^{f_{\bf X}(A)}, P^{A}({\cal A}) P^{f_{\bf X}(A)}_{\mix}) 
\\
\le &
\sqrt{M}
\sqrt{
\rE_{\bX}
d_2(P^{f_{\bf X}(A)}, P^{A}({\cal A}) P^{f_{\bf X}(A)}_{\mix}) 
}.
\end{align*}
Substituting 
(\ref{5-14-3}) and (\ref{5-14-4}) into the above inequality,
we obtain
\begin{align}
\rE_{\bX}
d_1(P^{f_{\bf X}(A)})
\le
M^{\frac{1}{2}} e^{-\frac{H_{2}(A|P^{A})}{2}}. \Label{5-14-6}
\end{align}

Using (\ref{5-14-6}), we can show the following theorem
as a generalization of (\ref{5-14-6}).
\begin{thm}\Label{thm1}
A family of universal$_2$ hash functions $f_{\bX}$ satisfies
\begin{align}
\rE_{\bX}
d_1(P^{f_{\bf X}(A)})
\le
3 M^{\frac{s}{1+s}} e^{-\frac{s H_{1+s}(A|P^{A})}{1+s}}
\hbox{ for } 0 \le s \le 1.
\Label{5-6-3}
\end{align}
\end{thm}

Substituting $s=1$, we obtain
\begin{align}
\rE_{\bX}
d_1(P^{f_{\bf X}(A)})
\le
3 M^{\frac{1}{2}} e^{-\frac{H_{2}(A|P^{A})}{2}}.\Label{5-14-7}
\end{align}
Since the difference between (\ref{5-14-6}) and (\ref{5-14-7}) is only the coefficient,
Theorem \ref{thm1} can be regarded as a kind of generalization of 
Bennett et al \cite{BBCM}'s result (\ref{5-14-4}).

\begin{proof}
For any $R'>0$, we choose the subset $\Omega_{R'}:=\{P^A(a) > e^{-R'}  \}$,
and define the sub-distribution $P^A_{R'}$ by
\begin{align*}
P^A_{R'}(a):=
\left\{
\begin{array}{ll}
0 & \hbox{ if } a \in \Omega_{R'} \\
P^A(a) & \hbox{ otherwise.}
\end{array}
\right.
\end{align*}
Since 
\begin{align*}
d_1(P^A,P^A_{R'})=P^A(\Omega_{R'})
\end{align*}
and
\begin{align*}
&d_1(P^{A}_{R'}({\cal A}) P^{f_{\bX}(A)}_{\mix}, P^{A}({\cal A}) P^{f_{\bf X}(A)}_{\mix}) \\
=&
d_1( 0 , 
(P^{A}({\cal A}) - P^{A}_{R'}({\cal A}))
P^{f_{\bf X}(A)}_{\mix}) \\
=&
(P^{A}({\cal A}) - P^{A}_{R'}({\cal A}))
d_1( 0 , P^{f_{\bf X}(A)}_{\mix}) \\
=& P^{A}({\cal A}) - P^{A}_{R'}({\cal A}) 
= P^A(\Omega_{R'}),
\end{align*}
the idea of ``smoothing'' by Renner \cite{Ren05} yields that
\begin{align}
& d_1(P^{f_{\bX}(A)})
= 
d_1(P^{f_{\bX}(A)}, P^{A}({\cal A}) P^{f_{\bf X}(A)}_{\mix}) 
\nonumber 
\\
\le &
d_1(P^{f_{\bX}(A)}, P^{f_{\bX}(A)}_{R'})
+d_1(P^{f_{\bX}(A)}_{R'}, P^{A}_{R'}({\cal A}) P^{f_{\bX}(A)}_{\mix} ) \nonumber \\
&+d_1(P^{A}_{R'}({\cal A}) P^{f_{\bX}(A)}_{\mix}, P^{A}({\cal A}) P^{f_{\bf X}(A)}_{\mix}) \nonumber \\
= &
2 P^A(\Omega_{R'})
+ d_1(P^{f_{\bX}(A)}_{R'}). 
\end{align}
Taking the expectation over $\bX$, we obtain 
\begin{align}
\rE_{\bX} d_1(P^{f_{\bX}(A)})
\le
2 P^A(\Omega_{R'})
+ \rE_{\bX} d_1(P^{f_{\bX}(A)}_{R'}). \Label{5-14-10}
\end{align}
The inequality (\ref{5-14-6}) yields 
\begin{align*}
\rE_{\bX} d_1(P^{f_{\bX}(A)}_{R'})
\le M^{\frac{1}{2}} e^{-\frac{1}{2}H_2(A|P^A_{R'})}.
\end{align*}
For $0 \le s \le 1$, we can evaluate $e^{-H_2(A|P^A_{R'})}$  and $P^A(\Omega_{R'})$ as
\begin{align}
& e^{-H_2(A|P^A_{R'})} = 
\sum_{a\in \Omega_{R'}^c} P^A(a)^2
\le
\sum_{a\in \Omega_{R'}^c} P^A(a)^{1+s} e^{-(1-s)R'} \nonumber \\
\le &
\sum_{a} P^A(a)^{1+s} e^{-(1-s)R'} 
=e^{-s H_{1+s}(A|P^A) -(1-s)R'} \Label{5-14-8}\\
& P^A(\Omega_{R'})  =
\sum_{a\in \Omega_{R'}} P^A(a)
\le 
\sum_{a\in \Omega_{R'}} (P^A(a))^{1+s} e^{sR'} \nonumber \\
\le & \sum_{a} (P^A(a))^{1+s} e^{sR'} 
=  e^{-s H_{1+s}(A|P^A)+ s R'} .\Label{5-14-9}
\end{align}
Combining (\ref{5-14-10}), (\ref{5-14-8}), and (\ref{5-14-9}),
for $R:=\log M$, we obtain
\begin{align*}
& \rE_{\bX} d_1(P^{f_{\bX}(A)})\\
\le &
2e^{-s H_{1+s}(A|P^A)+ s R'} 
+ e^{R+\frac{1}{2}(-s {H}_{1+s}(A|P^A) -(1-s)R')}\\
= &
3e^{-\frac{s{H}_{1+s}(A|P^A)+ s R}{1+s}},
\end{align*}
where we substitute $\frac{R+s {H}_{1+s}(A|P^A)}{1+s}$ into $R'$.
\end{proof}

Next, we consider the case when our distribution $P^{A_n}$ 
is given by the $n$-fold independent and identical distribution of 
$P^{A}$, i.e, $(P^{A})^n$.
When the random number generation rate $ \lim_{n \to \infty} \frac{1}{n}\log M_n$ is $R$,
we focus on the {\it exponential rate of decrease} of 
$\rE_{\bX} d_1(P^{f_{{\bX},n}(A_n)})$,
and consider the supremum.

When an ensemble $\{f_{{\bX},n} \}$ of hash functions is a family of universal$_2$ hash functions from $\cA^n$
to $\{1, \ldots M_n\}$,
Theorem \ref{thm1} yields that
\begin{align*}
& \liminf_{n\to\infty} \frac{-1}{n}\log 
\rE_{\bX} d_1(P^{f_{{\bX},n}(A_n)}) \nonumber \\
\ge &
\frac{s{H}_{1+s}(A|P^{A}) -s R}{1+s}
\end{align*}
for $s \in [0,1]$.
Taking the maximum over $s \in [0,1]$, we obtain
\begin{align}
& \liminf_{n\to\infty} \frac{-1}{n}\log 
\rE_{\bX} d_1(P^{f_{{\bX},n}(A_n)}) \nonumber \\
\ge &
\max_{0\le s \le 1} 
\frac{s{H}_{1+s}(A|P^{A}) -s R}{1+s}. \Label{4-19-2}
\end{align}
On the other hand, when we apply the Pinsker's inequality \cite{CKbook} to the upper bound for the mutual information
obtained by the previous paper \cite{Hayashi2},
we obtain another bound 
$\max_{0\le s \le 1} 
\frac{s{H}_{1+s}(A|P^{A}) -s R}{2}$, which is smaller than (\ref{4-19-2}).

\subsection{Protocol based on universal$_2$ hash functions: Converse part}
\Label{s2c}

In order to show the tightness of the exponential rate of decrease
(\ref{4-19-2}) under the universal$_2$ condition,
we consider the following property.

\begin{condition}[Strongly universal$_2$]\Label{C2}
For any $a\in \cA$,
${\rm Pr} \{ f_{\bX}(a)=m\} =\frac{1}{M}$.
The random variable
$f_{\bX}(a) $ is independent of $\{f_{\bX}(a')\}_{a' \neq a \in \cA}$.
% for different arbitrary two elements $a \neq a' \in \cA$.
\end{condition}

\begin{thm}\Label{thm2}
For any strongly universal$_2$ ensemble,
any subset $\Omega \subset \cA$ with $|\Omega | < M$
satisfies
\begin{align}
\rE_{\bX}
d_1(P^{f_{\bf X}(A)})
\ge
(1-\frac{|\Omega|}{M})^2 P^A(\Omega).
\Label{5-14-2}
\end{align}
\end{thm}
The proof is given in Appendix \ref{as2}.

In order to derive the inequality opposite to (\ref{4-19-2}) from Theorem \ref{thm2},
we employ the type method \cite{CKbook}.
In the type method,
when an $n$-trial data $\vec{a}_n:=(a_1,\ldots,a_n) \in {\cal A}^n$ is given,
we focus on the distribution $p(a):=\frac{\# \{i| a_i=a \}}{n}$,
which is called the empirical distribution for the data $\vec{a}_n$.
In the type method, an empirical distribution is called a type.
In the following, we denote the set of empirical distributions on $\cA$
with $n$ trials by ${\cal T}_n$.
The cardinality $|{\cal T}_n|$ is bounded by $(n+1)^{|{\cal A}|-1}$ \cite{CKbook},
which increases polynomially with the number $n$.
That is,
\begin{align}
\lim_{n \to \infty} \frac{1}{n}\log |{\cal T}_n|=0.
\Label{9-6-1}
\end{align}
This property is the key idea in the type method.
When $T_n(Q)$ represents the set of $n$-trial data whose empirical distribution is $Q$,
the cardinality of $T_n(Q)$ can be evaluated as \cite{CKbook}:
\begin{align}
\lceil \frac{e^{nH(Q)}}{|{\cal T}_n|} \rceil \le |T_n(Q)| \le \lfloor e^{nH(Q)} \rfloor,
\Label{9-7-5}
\end{align}
where $ \lceil x \rceil$ is the minimum integer $m$ satisfying $m \ge x $,
and $\lfloor x \rfloor$ is the maximum $m$ satisfying $m \le x $.
Since any element $\vec{a} \in T_n(Q)$ 
satisfies 
\begin{align}
P^{A_n}(\vec{a})=e^{-n (D(Q\|P^A)+H(Q)) }, \Label{9-7-13}
\end{align}
we obtain an important formula
\begin{align}
\frac{1}{{\cal T}_n} e^{-n D(Q\|P^A)}
\le P^{A_n}(T_n(Q)) \le 
e^{-n D(Q\|P^A)}
\Label{9-7-14}.
\end{align}
Using the above knowledge, 
we can show the following proposition:

\begin{proposition}
%When $R \ge R_c$,
When $M_n = \lfloor e^{n R} \rfloor$, any sequence of strongly universal$_2$ ensembles 
$\{f_{\bX,n}\}$ 
from $\cA^n$ to $\{1, \ldots M_n\}$
satisfies 
the equation
\begin{align}
& \limsup_{n\to\infty} \frac{-1}{n}\log 
\rE_{\bX} d_1(P^{f_{\bX,n}(A_n)}) 
\le 
%\max_{0\le s \le 1}  \frac{s (H_{1+s}(A|P^{A}) - R)}{1+s}= 
\min_{Q:H(Q)\le R} D(Q\|P^A),
\Label{9-7-2}
\end{align}
where $D(Q\|P^A)$ is the
Kullback-Leibler divergence $\sum_{a\in \cA} Q(a)(\log Q(a)-\log P^A(a))$.
\end{proposition}

\begin{proof}
Choose an arbitrary empirical distribution $Q\in {\cal T}_n$
satisfying that $H(Q) \le R$.
Then, due to (\ref{9-7-5}), the cardinality $|T_n(Q)|$ is less than
$\lfloor e^{nR}\rfloor$.
We choose the subset $\Omega_{n,Q}$ 
with the cardinality $\lceil  \frac{1}{2} e^{nR} \rceil $ so that
it contains at least $\lceil  \frac{|T_n(Q)|}{2} \rceil $ elements of $T_n(Q)$.
Using (\ref{9-7-5}) and (\ref{9-7-13}), we obtain
\begin{align*}
 P^{A_n}(\Omega_{n,Q})
\ge &
 \frac{|T_n(Q)|}{2}  e^{-n (D(Q\|P^A)+H(Q)) }  \\
\ge &
 \frac{e^{nH(Q)}}{2 |{\cal T}_n|} 
  e^{-n (D(Q\|P^A)+H(Q)) } .
\end{align*}
Using Theorem \ref{thm2} with $\Omega_{n,Q}$,
we obtain
\begin{align*}
& \rE_{\bX}
d_1(P^{f_{\bX,n}(A_n)}) 
\ge 
(1- \frac{\lceil  \frac{1}{2} e^{nR} \rceil}{\lfloor e^{nR}\rfloor})^2
\frac{1}{2|{\cal T}_n|}e^{-n D(Q\|P^A)}.
\end{align*}
Since $Q$ is an arbitrary empirical distribution $Q\in {\cal T}_n$
satisfying that $H(Q) \le R$,
\begin{align*}
& \rE_{\bX}
d_1(P^{f_{\bX,n}(A_n)}) \\
\ge &
(1- \frac{\lceil  \frac{1}{2} e^{nR} \rceil}{\lfloor e^{nR}\rfloor})^2
\frac{1}{2|{\cal T}_n|} 
\max_{Q \in {\cal T}_n: H(Q) \le R}
e^{-n D(Q\|P^A)}.
\end{align*}
That is,
\begin{align*}
& \frac{-1}{n}\log  \rE_{\bX}
d_1(P^{f_{\bX,n}(A_n)}) \\
\le &
\min_{Q \in {\cal T}_n: H(Q) \le R} D(Q\|P^A)
+ \frac{1}{n}\log 2|{\cal T}_n| \\
& - \frac{2}{n}\log(1- \frac{\lceil  \frac{1}{2} e^{nR} \rceil}{\lfloor e^{nR}\rfloor}).
\end{align*}
Due to the continuity of $Q \mapsto H(Q),D(Q\|P^A)$
and (\ref{9-6-1}),
the limit $n \to \infty$ yields (\ref{9-7-2}).
\end{proof}

When $R \le H(A|P^{A})$,
the equation
\begin{align}
\max_{0\le s}  \frac{s (H_{1+s}(A|P^{A}) - R)}{1+s}
= \min_{Q:H(Q)\le R} D(Q\|P^A)
\Label{9-7-1}
\end{align}
is known as the strong converse exponent in fixed source coding \cite{CKbook,Han-source},\cite[(A21)]{Haya4}.
The maximum $\max_{0\le s}  \frac{s (H_{1+s}(A|P^{A}) - R)}{1+s}$ is realized 
at $s=s_0$ when $R=R_{s_0}:= (1+s_0)\frac{d}{ds}(s H_{1+s}(A|P^{A}) )|_{s=s_0}- s_0 H_{1+s_0}(A|P^{A})$.
Since $\frac{d}{ds}R_s= (1+s)\frac{d^2}{ds^2}(s H_{1+s}(A|P^{A}) )\le 0$,
$R_s$ is monotone decreasing with $s$.

Thus, when $H(A|P^{A}) \ge R \ge R_1$ ($R_1$ is called the critical rate.), 
\begin{align}
\max_{0\le s}  \frac{s (H_{1+s}(A|P^{A}) - R)}{1+s}
=
\max_{0\le s\le 1}  \frac{s (H_{1+s}(A|P^{A}) - R)}{1+s}.
\Label{9-7-3}
\end{align}
Hence, in this case,
due to (\ref{4-19-2}), (\ref{9-7-2}), (\ref{9-7-1}), and (\ref{9-7-3}),
we obtain
\begin{align}
& \lim_{n\to\infty} \frac{-1}{n}\log 
\rE_{\bX} d_1(P^{f_{\bX,n}(A_n)}) \nonumber \\
= & 
\max_{0\le s \le 1}  \frac{s (H_{1+s}(A|P^{A}) - R)}{1+s}= 
\min_{Q:H(Q)\le R} D(Q\|P^A).
\Label{9-7-6}
\end{align}
However, when $R < R_1$,
\begin{align*}
&\max_{0\le s\le 1}  \frac{s (H_{1+s}(A|P^{A}) - R)}{1+s}
=
\frac{H_{2}(A|P^{A}) - R}{2} \\
<&
\max_{0\le s }  \frac{s (H_{1+s}(A|P^{A}) - R)}{1+s}.
\end{align*}
So, the lower bound in (\ref{4-19-2}) does not coincide with the upper bound in (\ref{9-7-2}).

\subsection{Comparison with evaluation by Holenstein-Renner \cite{Holenstein}}
\Label{s2b}
In the above derivation, 
the key point is evaluating the probability $P^A(\Omega_{R'})$,
which equals the probability $(P^A)^n\{ a\in{\cal A}^n|(P^A)^n(a) > e^{-n R'} \}$ in the $n$-i.i.d. setting.
In the community of cryptography,
the $n$-i.i.d. setting 
is not regarded as an important setting
because they are more interested in the single-shot setting.
In such a setting,
they sometimes use Holenstein-Renner's \cite{Holenstein} evaluation of $P^X(\Omega_{R'})$.
They proved the following theorem.
\begin{thm}\Label{Holenstein}
When $0 \le H(A)-R' \le \log |{\cal A}|$,
\begin{align}
(P^A)^n\{ a\in{\cal A}^n| (P^A)^n(a) > e^{-n R'} \}
\le 2^{- \frac{n (H(A)-R')^2}{2 (\log (|{\cal A}|+3))^2}}.
\end{align}
Further,
when 
$|{\cal A}| \ge 3$ and 
$0\le H(A)-R' \le \frac{\log (|{\cal A}|-1)}{12}$,
\begin{align*}
(P^A)^n\{ a\in{\cal A}^n| (P^A)^n(a) > e^{-n R'} \}
> \frac{1}{110} 2^{- \frac{12 n (H(A)-R')^2}{(\log (|{\cal A}|-1))^2}}.
%\Label{2-3-1}
\end{align*}
When $|{\cal A}| =2$,
the inequality yields the following evaluation.
When 
$0\le H(A)-R' \le \frac{\log 3}{24}$,
\begin{align*}
(P^A)^n\{ a\in{\cal A}^n| (P^A)^n(a) > e^{-n R'} \}
> \frac{1}{110} 2^{- \frac{24 n (H(A)-R')^2}{(\log 3)^2}}
%\Label{2-3-2}
\end{align*}
for even $n$.
\end{thm}

Our evaluation (\ref{5-14-9}) of 
$(P^A)^n\{ a\in{\cal A}^n|(P^A)^n(a) > e^{-n R'} \}$
contains the parameter $0\le s\le 1$.
Since this parameter is arbitrary,
it is natural to compare the upper bound $\min_{0 \le s \le 1}  e^{-n (s {H}_{1+s}(X|P^X)- s R')}$ given by 
(\ref{5-14-9}) with that by Theorem \ref{Holenstein}.
That is, using (\ref{5-14-9}), we obtain the exponential evaluation
\begin{align*}
& \lim_{n \to \infty} \frac{-1}{n}\log (P^A)^n\{ a\in{\cal A}^n| (P^A)^n(a) > e^{-n R'} \} \\
\ge  & \max_{0 \le s } s {H}_{1+s}(A|P^A)- s R',
\end{align*}
while 
Theorem \ref{Holenstein} yields that
\begin{align*}
& \lim_{n \to \infty} \frac{-1}{n}\log (P^A)^n\{ a\in{\cal A}^n| (P^A)^n(a) > e^{-n R'} \}\\
\ge & \frac{ (H(A)-R')^2}{2 (\log (|{\cal A}|+3))^2}\log 2 .
\end{align*}
In this case, the upper bound is
$\frac{12 \log 2 (H(A)-R')^2}{(\log (|{\cal A}|-1))^2}$ for $|{\cal A}| \ge 3$
and 
$\frac{24 \log 2 (H(A)-R')^2}{(\log 3)^2}$ for $|{\cal A}| = 2$.

In fact, the probability $P^A(\Omega_{R'})$ is the key quantity in the method of information spectrum, 
which is a unified method in information theory \cite{Han}.
When the method of information spectrum is applied to an i.i.d. source,
the probability $P^A(\Omega_{R'})$ is evaluated by applying Cram\'{e}r's Theorem (see \cite{DZ}) to the random variable $\log P^A(a)$.
Then we obtain 
\begin{align}
& \lim_{n \to \infty} \frac{-1}{n}\log (P^A)^n\{ a\in{\cal A}^n| (P^A)^n(a) > e^{-n R'} \}\nonumber\\
= & \max_{0 \le s } s {H}_{1+s}(A|P^A)- s R'\Label{9-9-10}
\end{align}
for $R \le H(A)$.
Since $s \mapsto s {H}_{1+s}(X|P^X)$ is concave,
when $H(A) \ge R \ge H_2'(A|P^A)$,
the maximization (\ref{9-9-10}) can be attained with $s\in [0,1]$, i.e.,
\begin{align*}
& \lim_{n \to \infty} \frac{-1}{n}\log (P^A)^n\{ a\in{\cal A}^n| (P^A)^n(a) > e^{-n R'} \} \\
= & \max_{0 \le s \le 1} s {H}_{1+s}(A|P^A)- s R'.
\end{align*}
which implies that
our evaluation (\ref{5-14-9}) gives the tight bound for exponential rate of decrease for 
the probability $(P^A)^n\{ a\in{\cal A}^n| (P^A)^n(a) > e^{-n R'} \}$.
In fact, the difference among these bounds is numerically given in Fig. \ref{f2}.
Therefore, we can conclude that
our evaluation (\ref{5-14-9}) is much better than that by Holenstein-Renner \cite{Holenstein}. 
That is, the combination of Lemma 1 and (\ref{5-14-9})
is essential for deriving the tight exponential bound.

\begin{figure}[htbp]
\begin{center}
\scalebox{1.0}{\includegraphics[scale=0.8]{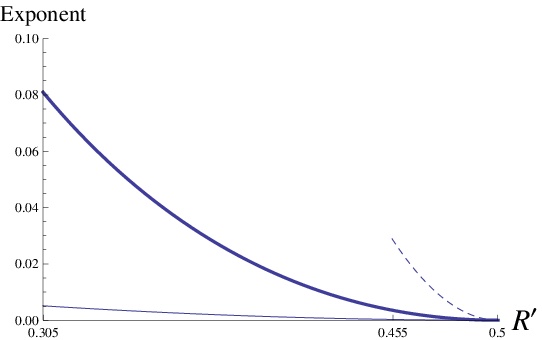}}
\end{center}
\caption{
Evaluation of
$\lim_{n \to \infty} \frac{-1}{n}\log (P^A)^n\{a\in{\cal A}^n| (P^A)^n(a) > e^{-n R'} \}$.
Thick line: $\max_{0 \le s \le 1} s {H}_{1+s}(A|P^A)- s R'$ (The present paper),
Normal line: $\frac{ (H(A)-R')^2}{2 (\log (|{\cal A}|+3))^2}\log 2 $ 
(Lower bound by \cite{Holenstein}),
Dashed line: $\frac{24 \log 2 (H(A)-R')^2}{(\log 3)^2}$ 
(Upper bound by \cite{Holenstein}).
Here,
$P^A$ is chosen to be the binary distribution 
$P^A(0)= \alpha$, $P^A(1)=1-\alpha$ with $\alpha=0.200$.
Then, $h(\alpha)=H(A)=0.500$,
$\frac{d (sH_{1+s}(A))}{ds}|_{s=1}=0.305$, and $H(A)- \frac{\log 3}{24}=0.455$.}
\Label{f2}
\end{figure}%

\subsection{Comparison with smooth min-entropy}\Label{s2d}
In subsection \ref{s2a}, we treated smoothing of R\'{e}nyi entropy of order $2$.
In this subsection, we compare this method with 
smooth min-entropy, which is more familiar in the community of cryptography \cite{Ren05}.
When we employ the min-entropy 
$H_{\min}(A|P^{A}):= -\log \max_{a\in A}P^A(a)$
instead of R\'{e}nyi entropy of order $2$ in (\ref{5-14-6}), 
we obtain the following inequality:
\begin{align}
\rE_{\bX}
d_1(P^{f_{\bf X}(A)})
\le
M^{\frac{1}{2}} e^{-\frac{H_{\min}(A|P^{A})}{2}}. 
\Label{5-14-6-2}
\end{align}
Now, we choose another distribution $\tilde{P}^A$ satisfying $d_1(\tilde{P}^A,P^A) \le \epsilon$.
Using 
(\ref{1-5-4}), (\ref{5-14-6-2}), and 
$\epsilon$-smooth min-entropy
$H_{\min,\epsilon}(A|P^{A}):=\max_{P:d_1(\tilde{P}^A,P^A)\le \epsilon} H_{\min}(A|\tilde{P}^A)$,
we can show the following inequality \cite{Ren05}
\begin{align}
&\rE_{\bX}
d_1(P^{f_{\bf X}(A)}) 
=
\rE_{\bX}
d_1(P^{f_{\bf X}(A)}, P^{A}({\cal A}) P^{f_{\bf X}(A)}_{\mix}) \nonumber \\
\le &
\rE_{\bX}
d_1(\tilde{P}^{f_{\bf X}(A)}, \tilde{P}^{A}({\cal A}) P^{f_{\bf X}(A)}_{\mix}) 
+
d_1(\tilde{P}^{f_{\bf X}(A)}, P^{f_{\bf X}(A)}) \nonumber \\
&+
d_1(\tilde{P}^{A}({\cal A}) P^{f_{\bf X}(A)}_{\mix}, P^{A}({\cal A}) P^{f_{\bf X}(A)}_{\mix}) 
\nonumber  \\
\le &
M^{\frac{1}{2}} e^{-\frac{H_{\min}(A|\tilde{P}^{A})}{2}}
+
d_1(\tilde{P}^{A}(A), P^{A}(A)) \nonumber \\
& +
|\tilde{P}^{A}({\cal A}) P^{f_{\bf X}(A)}_{\mix}, P^{A}({\cal A}) | 
\nonumber \\
\le &
M^{\frac{1}{2}} e^{-\frac{H_{\min,\epsilon}(A|P^{A})}{2}}
+2\epsilon.
\Label{5-14-6-3}
\end{align}
Next, 
using the subdistribution $P^A_{R'}$ defined in proof of Theorem \ref{thm1},
we choose $\epsilon$ to be
$d_1(P^A,P^A_{R'})=P^A(\Omega_{R'})$
for a given $R'\ge \log M$.
Then, 
\begin{align}
&\rE_{\bX}
d_1(P^{f_{\bf X}(A)})
\le  
M^{\frac{1}{2}} e^{-\frac{H_{\min,\epsilon}(A|P^{A})}{2}}
+2\epsilon \nonumber \\
\le &
M^{\frac{1}{2}} e^{-\frac{H_{\min}(A|P^{A}_{R'})}{2}}
+2\epsilon  \nonumber \\
\le &
\sqrt{\frac{M}{e^{R'}}}
+2 
P^A
\{
a \in {\cal A}|
P^A(a) > e^{-R'}
\}
%P^A(\Omega_{R'}) 
\Label{speLem} .
\end{align}
Applying the inequality (\ref{5-14-9}),
we obtain
\begin{align*}
\rE_{\bX}
d_1(P^{f_{\bf X}(A)})
\le 
M^{\frac{1}{2}} e^{-R'/2}
+2 e^{-s H_{1+s}(A|P^A)+ s R'} 
\end{align*}
for $s \ge 0$.
When $R=\log M$,
\begin{align}
&\rE_{\bX}
d_1(P^{f_{\bf X}(A)})
\le
e^{(R-R')/2}+2 e^{-(s H_{1+s}(A|P^A)-s R')}.
\end{align}

Now, we choose $R'=R'_0$ such that
$(R'_0-R)/2=
s H_{1+s}(A|P^A)-s R'_0$,
which implies
$
R'_0= \frac{R+2 s H_{1+s}(A|P^A)}{1+2 s}$.
Hence,
$(R'_0-R)/2=\frac{s H_{1+s}(A|P^A)-s R }{1+2 s}$.
Thus, we obtain
\begin{align}
&\rE_{\bX}
d_1(P^{f_{\bf X}(A)})
\le
3 e^{- \frac{s H_{1+s}(A|P^A)-s R}{1+2s}}.
\end{align}
Taking the minimum over $s >0$, we have
\begin{align}
&\rE_{\bX}
d_1(P^{f_{\bf X}(A)})
\le
3 e^{- \max_{s\ge 0} \frac{s H_{1+s}(A|P^A)-s R}{1+2s}}.\Label{4-19-2-2}
\end{align}

Next, we consider the case when our distribution $P^{A_n}$ 
is given by the $n$-fold independent and identical distribution of 
$P^{A}$, i.e, $(P^{A})^n$.
Similar to (\ref{4-19-2}), (\ref{4-19-2-2}) yields
\begin{align}
& \liminf_{n\to\infty} \frac{-1}{n}\log 
\rE_{\bX} d_1(P^{f_{{\bX},n}(A_n)}) \nonumber \\
\ge &
\max_{0\le s } 
\frac{s{H}_{1+s}(A|P^{A}) -s R}{1+2s} \nonumber\\
=&
\max_{0\le t \le 1 } 
\frac{t{H}_{1/(1-t)}(A|P^{A}) -t R}{1+t}, \Label{4-19-2-1}
\end{align}
where $t=\frac{s}{1+s}$.
In fact, as shown in Appendix \ref{8-03-e}, 
the exponential decreasing rate of the right hand side of (\ref{5-14-6-3})
is calculated as
\begin{align}
& \lim_{n\to\infty} \frac{-1}{n}\log 
\min_{\epsilon}(e^{\frac{nR}{2}} e^{-\frac{H_{\min,\epsilon}(A^n|(P^{A})^n)}{2}}
+2\epsilon )
\nonumber \\
=&
\max_{0\le t \le 1 } 
\frac{t{H}_{1/(1-t)}(A|P^{A}) -t R}{1+t}. \Label{8-03-a}
\end{align}
Hence, we can consider that 
$\max_{0\le t \le 1 } 
\frac{t{H}_{1/(1-t)}(A|P^{A}) -t R}{1+t}$
expresses the optimal exponential decreasing rate for the method of smooth min-entropy.
For $0\le t \le 1$,
the relation $t \le \frac{t}{1-t}$ implies the inequality
${H}_{1/(1-t)}(A|P^{A}) \le {H}_{1+t}(A|P^{A})$.
Hence,
the bound
$\max_{0\le s } \frac{s{H}_{1+s}(A|P^{A}) -s R}{1+2s}$
is smaller than
the presented bound
$\max_{0\le s \le 1} \frac{s{H}_{1+s}(A|P^{A}) -s R}{1+s}$,
whose numerical comparison is illustrated in Fig. \ref{f20}.

\begin{figure}[htbp]
\begin{center}
\scalebox{1.0}{\includegraphics[scale=0.8]{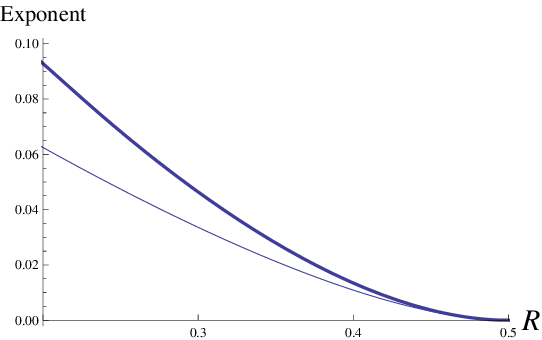}}
\end{center}
\caption{
Comparison between
$\max_{0\le s } \frac{s{H}_{1+s}(A|P^{A}) -s R}{1+s}$
and
$\max_{0\le s } \frac{s{H}_{1+s}(A|P^{A}) -s R}{1+2s}$.
Thick line: 
$\max_{0\le s } \frac{s{H}_{1+s}(A|P^{A}) -s R}{1+s}$
(Smoothing of R\'{e}nyi entropy of order $2$. The present paper),
Normal line: 
$\max_{0\le s } \frac{s{H}_{1+s}(A|P^{A}) -s R}{1+2s}$.
(Smoothing of min entropy. \cite{Ren05}).
Here,
$P^A$ is chosen to be the binary distribution 
$P^A(0)=\alpha$, $P^A(1)=1-\alpha$ with $\alpha=0.200$.
Then, $h(\alpha)=H(A)=0.500$.}
\Label{f20}
\end{figure}%

\section{Specialized protocol for uniform random number generation}
\Label{s3}
\subsection{Main result of this section}
Next, we consider 
a function $f$ from $\cA$ to $\{1, \ldots, M\}$
specialized to a given probability distribution $P^A$.
This problem is called intrinsic randomness, 
which was studied with general source by 
Vembu and Verd\'{u} \cite{VV}.
The previous paper \cite{H08} discussed the relation between the second order asymptotic rate and 
the central limit theorem.
In the following, for the comparison with the exponential rate of decrease for (\ref{9-7-1}),
we prove the following theorem,
which gives the optimal exponential rate of decrease  
for a given rate of uniform random number generation.

\begin{thm}\Label{t9-9-1}
When $\frac{d (s{H}_{1+s}(A|P))}{ds}|_{s=1} \le R $, 
we obtain
\begin{align}
& \lim_{n\to\infty} \frac{-1}{n}
\log \min_{f_n\in {\cal F}_n(R) } d_1(P^{f_n(A_n)}) \nonumber \\
=& \max_{0\le s \le 1}  s (H_{1+s}(A|P^{A}) - R) ,
\Label{9-7-10}
\end{align}
where
${\cal F}_n(R)$ is the set of functions $f_n$ from ${\cal A}^n$ to
$\{1, \ldots, \lfloor e^{nR} \rfloor \}$.
\end{thm}

Combining (\ref{9-7-6}) and Theorem \ref{t9-9-1},
we can compare the performances between 
a random universal protocol and the best specialized protocol.
So, our exponential rate of decrease for the protocol based on universal$_2$ hash functions  
is slightly smaller than 
the optimal exponential rate of decrease for specialized protocols.

In order to prove Theorem \ref{t9-9-1}, 
we will show the following two inequalities:
\begin{align}
& \limsup_{n\to\infty} \frac{-1}{n}
\log \min_{f_n\in {\cal F}_n(R) } d_1(P^{f_n(A_n)}) \nonumber \\
\le & \max_{0\le s \le 1}  s (H_{1+s}(A|P^{A}) - R)  
\Label{1-5-1} \\
& \liminf_{n\to\infty} \frac{-1}{n}
\log \min_{f_n\in {\cal F}_n(R) } d_1(P^{f_n(A_n)}) \nonumber \\
\ge & \max_{0\le s \le 1}  s (H_{1+s}(A|P^{A}) - R)  
\Label{1-5-2} .
\end{align}
Inequality (\ref{1-5-1}) is called the converse part
and 
Inequality (\ref{1-5-2}) is called the direct part
in the information theory community.
In order to show the respective inequalities,
we prepare respective lemmas (Lemmas \ref{LemDI} and \ref{LemCO})
in the non-asymptotic setting 
in Subsection \ref{s32}.
In Subsection \ref{s33},
using Lemma \ref{LemCO}
and the concavity property, 
we show the converse part (\ref{1-5-1}).
Also,
using Lemma \ref{LemDI}, we show 
the direct part (\ref{1-5-1}).
In the latter derivation,
we employ again the method of types \cite{CKbook}.

\subsection{Non-asymptotic evaluation}\Label{s32}
In order to treat the non-asymptotic case,
we introduce the notation:
\begin{align*}
[x]_+&:=\left\{
\begin{array}{ll} 
x & \hbox{ if } x \ge 0 \\
0 & \hbox{ if } x < 0 .
\end{array} 
\right.
\end{align*}
Then, the $L_1$ norm for two 
normalized distributions $P$ and $Q$ can be simplified to
\begin{align}
\sum_{a}|P(a)-Q(a)|
=
2 \sum_{a}[P(a)-Q(a)]_+,
\end{align}
which is a useful formula for the following discussion.

Hence, we obtain
the following lemma, which is 
useful for our proof of the direct part (\ref{1-5-2}).
\begin{lem}\Label{LemDI}
Any probability distribution $P^A$
and any function $f$ from $\cA$ to $\{1, \ldots, M\}$
satisfy that
\begin{align}
 d_1(P^{f(A)}) 
\ge 
P^A \{a\in{\cal A}| P^A(a)\ge \frac{2}{M}\}.\Label{9-8-1}
\end{align}
\end{lem}
\begin{proof}
%When $f^{-1}\{i\}\cap \{ a|P^A(a)\ge \frac{2}{M}\} $ is not empty,
%\begin{align*}
%P^A(f^{-1}\{i \})-\frac{1}{M}
%\ge 
%\sum_{a \in f^{-1}\{i\}\cap \{a| P^A(a)\ge \frac{2}{M}\} }
%\frac{P^A(a)}{2}.
%\end{align*}

Any positive numbers $\alpha_1,\ldots, \alpha_k$
satisfies
\begin{align}
[\sum_{i=1}^k \alpha_i -\frac{1}{M}]_+
\ge
\sum_{i=1}^k 
[\alpha_i -\frac{1}{M}]_+ \Label{9-10-1}.
\end{align}
When $P^A(a)\ge \frac{2}{M}$,
$P^A(a)-\frac{1}{M}\ge \frac{1}{M}$,
which implies that 
\begin{align}
& 2[P^A(a)-\frac{1}{M}]_+ 
=
2(P^A(a)-\frac{1}{M}) \nonumber \\
\ge &
P^A(a)-\frac{1}{M}+\frac{1}{M}= P^A(a).
\Label{9-10-2}
\end{align}
Thus, we obtain
\begin{align}
& \sum_{b}
|P^A(f^{-1}(b ))-\frac{1}{M}| 
=  
2 \sum_{b}
[P^A(f^{-1}(b))-\frac{1}{M}]_+ \nonumber \\
\ge &
2 \sum_{a\in{\cal A}}
[P^A(a)-\frac{1}{M}]_+ \Label{9-10-3}\\
\ge &
2 
\sum_{a\in{\cal A}: P^A(a)\ge \frac{2}{M}}
[P^A(a)-\frac{1}{M}]_+ \nonumber \\
\ge &
\sum_{a\in{\cal A}: P^A(a)\ge \frac{2}{M}}
P^A(a),\Label{9-10-4}
\end{align}
where (\ref{9-10-3}) and (\ref{9-10-4}) follows from (\ref{9-10-1}) and (\ref{9-10-2}).
Therefore, we obtain (\ref{9-8-1}).
\end{proof}

In order to show the converse part, 
we prepare the following lemma.
\begin{lem}\Label{Lem9-9}
Assume that 
for two integers $M \ge N$,
two positive number sequences $\alpha_1, \ldots, \alpha_{N}$
and $\beta_1,\ldots,\beta_{M}$ satisfy that
$\sum_{i=1}^N \alpha_i \ge
\sum_{i=1}^M \beta_i $.
Then, there exists a map $f$ from $\{1, \ldots, M\}$ to $\{1, \ldots, N\}$
such that
\begin{align}
\sum_{i=1}^{N}
[\sum_{j\in f^{-1}(i)}\beta_j - \alpha_i ]_+
\le N \max_j \beta_j .
\Label{9-9-3}
\end{align}
%where $[x]_+:=\max\{x,0\}$.
\end{lem}
\begin{proof}
First, we define $f(1):=1$.
For $j>1$, we define $f(j)$ inductively.
When $\sum_{j'\in f^{-1}(f(j-1))} \beta_{j'} < \alpha_{f(j-1)}$,
we define $f(j):=f(j-1)$.
Otherwise, we define $f(j):=f(j-1)+1$.
Then the function satisfies the condition (\ref{9-9-3}).
\end{proof}

Now we consider the case when our distribution $P^{A_n}$ 
is given by the $n$-fold independent and identical distribution of 
$P^{A}$, i.e, $(P^{A})^n$.
Using Lemma \ref{Lem9-9}, we have the following lemma,
which is 
useful for our proof of the converse part (\ref{1-5-1}).

\begin{lem}\Label{LemCO}
For any probability distribution $P^A$,
there exists a function $f_n$ from $\cA^n$ to $\{1, \ldots, M_n\}$
such that
\begin{align}
& d_1(P^{f_n(A_n)}) \nonumber \\
\le &
2(P^A)^n \{ a\in{\cal A}^n| (P^A)^n(a)\ge \frac{1}{M_n}\}\nonumber \\
&+ 2\sum_{Q\in {\cal T}_n^1[M_n]} M_n 
e^{-n (D(Q\|P^A) + H(Q))} 
\cdot
(P^A)^n(T_n(Q)) 
\nonumber \\
&+ 2 |{\cal T}_n| \max_{Q\in {\cal T}_n^2[M_n]} e^{-n (D(Q\|P^A) + H(Q))} 
\Label{9-9-1}
\end{align}
where 
\begin{align*}
{\cal T}_n^1[M_n] & :=\{ Q \in {\cal T}_n| D(Q\| P^A)+H(Q) \ge \frac{1}{n}\log M_n \} \\
{\cal T}_n^2[M_n] & :=\{ Q \in {\cal T}_n| (P^A)^n(T_n(Q)) < \frac{1}{M_n} \} .
\end{align*}
\end{lem}
\begin{proof}
In the first step, we define the function $f_n$.
In the second step, we show that the function satisfies (\ref{9-9-1}).

We divide ${\cal T}_n$ into three parts:
\begin{align*}
\tilde{\cal T}_n^0[M_n]&:=\{ Q\in {\cal T}_n|e^{n (D(Q\|P^A)+H(Q))} \le M_n\} \\
\tilde{\cal T}_n^1[M_n]&:=\{ Q\in (\tilde{\cal T}_n^0[M_n])^c
\cap {\cal T}_n |(P^A)^n(T_n(Q)) \ge  \frac{1}{M_n} \} \\
\tilde{\cal T}_n^2[M_n]&:=\{ Q\in (\tilde{\cal T}_n^0[M_n])^c
\cap {\cal T}_n |(P^A)^n(T_n(Q)) < \frac{1}{M_n} \} ,
\end{align*}
where $(\tilde{\cal T}_n^0[M_n])^c$ is the complement of $\tilde{\cal T}_n^0[M_n]$.
These three parts have the following relation with the above two parts:
\begin{align*}
\tilde{\cal T}_n^1[M_n]\subset {\cal T}_n^1[M_n] ,\quad
\tilde{\cal T}_n^2[M_n]\subset {\cal T}_n^2[M_n] .
\end{align*}

By using the integer $n_Q:=\lfloor \frac{(P^A)^n(T_n(Q))}{1/M_n}
\rfloor=\lfloor M_n (P^A)^n(T_n(Q))\rfloor$,
the conditions for $
\tilde{\cal T}_n^1[M_n]$ and $\tilde{\cal T}_n^2[M_n]$ are written as
$n_Q \ge 1$ and $n_Q <1$, respectively.
Note that, since $n_Q$ is a non-negative integer,
$n_Q <1$ is equivalent to $n_Q =0$.
%where $ \lfloor x \rfloor$ is the maximum integer $m$ satisfying $x \ge m$.

Due to (\ref{9-7-13}),
the condition that $e^{n (D(Q\|P^A)+H(Q))} \le M_n$ is equivalent with the condition that
$P^{A_n}(a) \ge \frac{1}{M_n}$ for $a \in T_n(Q)$.
Hence,
\begin{align}
 (P^A)^n \{ a\in{\cal A}^n| (P^A)^n(a)\ge \frac{1}{M_n}\}
= \sum_{Q \in {\cal T}_n^0} (P^A)^n(T_n(Q))\Label{2-2-6}.
\end{align}
So, 
\begin{align*}
& (P^A)^n \{ a\in{\cal A}^n| (P^A)^n(a)\ge \frac{1}{M_n}\}
%\sum_{Q \in {\cal T}_n^0} (P^A)^n(T_n(Q))
+
\sum_{Q \in \tilde{\cal T}_n^1[M_n]} \frac{n_Q}{M_n} \\
\le &
%(P^A)^n \{ (P^A)^n(a)\ge \frac{1}{M_n}\}
\sum_{Q \in \tilde{\cal T}_n^0[M_n]} (P^A)^n(T_n(Q))
+
\sum_{Q \in \tilde{\cal T}_n^1[M_n]} (P^A)^n(T_n(Q))
\le 1.
\end{align*}
Since 
\begin{align*}
& \frac{1}{M_n} \sum_{Q \in \tilde{\cal T}_n^0[M_n]} |T_n(Q)| 
= 
\frac{1}{M_n} |\{ a\in{\cal A}^n| (P^A)^n(a)\ge \frac{1}{M_n}\}| \\
\le & (P^A)^n \{ a\in{\cal A}^n|(P^A)^n(a)\ge \frac{1}{M_n}\}  ,
\end{align*}
we have
\begin{align*}
\sum_{Q \in \tilde{\cal T}_n^0[M_n]} |T_n(Q)| 
+\sum_{Q \in \tilde{\cal T}_n^1[M_n]} n_Q
\le M_n.
\end{align*}
Therefore, 
we can choose $f_n'$ on $\Omega':=
\cup_{Q\in \tilde{\cal T}_n^0[M_n] \cup \tilde{\cal T}_n^1[M_n]}T_n(Q)$
satisfying the following conditions.
\begin{enumerate}%[(1)]
\item
For $Q,Q'\in \tilde{\cal T}_n^0[M_n] \cup \tilde{\cal T}_n^1[M_n]$, 
$f_n'(T_n(Q))\cap f_n'(T(Q'))=\emptyset $.
\item
$f_n'|_{T_n(Q)}$ is injective for $Q\in \tilde{\cal T}_n^0[M_n]$.
\item
$|f_n'(T_n(Q))|=n_Q$ for $Q\in \tilde{\cal T}_n^1[M_n]$.
%Further, we choose $f_n'$ satisfying the additional condition.
\item
Any type $Q\in \tilde{\cal T}_n^1[M_n]$ satisfies that $|{f_n'}^{-1}(b)| \le \frac{|T_n(Q)|}{n_Q}$ for $b \in f_n'(T_n(Q))$.
\end{enumerate}
Then, for $Q\in \tilde{\cal T}_n^1[M_n]$, we obtain
\begin{align}
P^{f_n'(A_n)} (b) \le \frac{1}{M_n} + e^{-n (D(Q\|P^A)+H(Q))},
~\forall 
b \in f_n'(T_n(Q)).
\Label{2-2-1}
\end{align}
From the construction, 
\begin{align*}
\sum_{b \in f_n'(\Omega')} 
P^{f_n'(A_n)} (b) 
\ge \frac{1}{M_n} |f_n'(\Omega')|.
\end{align*}
That is, 
\begin{align}
\sum_{a \in (\Omega')^c} P^{A_n }(a)
\le
\frac{1}{M_n} | (f_n'(\Omega'))^c|.
\Label{9-9-7}
\end{align}

Next, we define $f_n$ on the whole set
by modifying $f_n'$ as follows.
\begin{enumerate}%[(5)]
\setcounter{enumi}{4}
\item
$f_n$ is the same as $f_n'$ on $\Omega'$.
\item
Due to (\ref{9-9-7}), we can apply Lemma \ref{Lem9-9} to the case when
$\{1, \ldots, N\}= (f_n'(\Omega'))^c$,
$\{1, \ldots, M\}= (\Omega')^c$,
$\alpha_{b}=\frac{1}{M_n}$ for $b \in (f_n'(\Omega'))^c$
and $\beta_a=P^{A_n}(a)$ for $a \in (\Omega')^c$.
Following Lemma \ref{Lem9-9},
we define the map $f_n|_{(\Omega')^c}$ from $(\Omega')^c$ to $(f_n'(\Omega'))^c$.
\end{enumerate}

Our remaining task is to evaluate the value $\sum_{b} [P^{f_n(A_n)}(b)-\frac{1}{M_n}]_+$.
Now, we define 
\begin{align*}
C(Q)&:=
\sum_{b\in f_n(T_n(Q))}
[P^{f_n(A_n)}(b)-\frac{1}{M_n}]_+.
\end{align*}
Then, (\ref{2-2-6}) implies that 
\begin{align}
\sum_{Q\in \tilde{\cal T}_n^0[M_n]} C(Q)
\le
(P^A)^n \{ a\in{\cal A}^n| (P^A)^n(a)\ge \frac{1}{M_n}\}.\Label{2-2-2}
\end{align}
For $Q\in  \tilde{\cal T}_n^1[M_n]$, 
(\ref{2-2-1}) implies
\begin{align}
C(Q) \le &
n_Q e^{-n D(Q\|P^A) -n H(Q)}  \nonumber \\
\le &
M_n e^{-n D(Q\|P^A) -n H(Q)} 
\cdot (P^A)^n(T_n(Q)) 
.\Label{2-2-3}
\end{align}
Thus, (\ref{2-2-2}) and (\ref{2-2-3}) imply
\begin{align}
& \sum_{b \in f_n'(\Omega') } [P^{f_n(A_n)}(b)-\frac{1}{M_n}]_+ \nonumber \\
\le &
(P^A)^n \{ a\in{\cal A}^n| (P^A)^n(a)\ge \frac{1}{M_n}\}\nonumber \\
&+ \sum_{Q\in {\cal T}_n^1[M_n] } 
M_n e^{-n D(Q\|P^A) -n H(Q)}
\cdot
(P^A)^n(T_n(Q)) 
\Label{9-9-6}.
\end{align}
Recall the condition 6). Lemma \ref{Lem9-9} guarantees that
\begin{align}
& \sum_{b \in (f_n'(\Omega'))^c } [P^{f_n(A_n)}(b)-\frac{1}{M_n}]_+ \nonumber \\
\le &
|(f_n'(\Omega'))^c | \max_{Q \in \tilde{\cal T}_n^2[M_n]} 
e^{-n (D(Q\|P^A) + H(Q))} \nonumber \\
\le &
|{\cal T}_n | \max_{Q \in {\cal T}_n^2[M_n]} e^{-n (D(Q\|P^A) + H(Q))} .
\Label{2-2-4}
\end{align}
Combining (\ref{9-9-6}) and (\ref{2-2-4}), we obtain
(\ref{9-9-1}).
\end{proof}

\subsection{Asymptotic evaluation}\Label{s33}
Next, we proceed to the asymptotic evaluation.
First, using Cram\'{e}r's Theorem \cite{DZ}, we obtain
\begin{align}
&\max_{0\le s }  s {H}_{1+s}(A|P^{A}) -s R \nonumber \\
=&\lim_{n \to \infty} \frac{-1}{n}\log 
(P^A)^n \{ a\in{\cal A}^n| (P^A)^n(a)\ge \frac{1}{e^{nR}}\} 
\Label{9-7-11}
%\\
%=&\lim_{n \to \infty} \frac{-1}{n}\log 
%(P^A)^n \{ a\in{\cal A}^n| (P^A)^n(a)\ge \frac{2}{e^{nR}}\}.
%\Label{9-7-11-b}
\end{align}
Hence, Equality (\ref{9-7-11}) and Lemma \ref{LemDI} imply
\begin{align}
& \limsup_{n\to\infty} \frac{-1}{n}\log \min_{f_n \in {\cal F}_n(R) } d_1(P^{f_n(A_n)}) \nonumber \\
\le & \max_{0\le s }  s (H_{1+s}(A|P^{A}) - R) 
\Label{9-7-10-a}.
\end{align}
Since $s \mapsto s H_{1+s}(A|P^{A})$ is concave, 
when $\frac{d (s{H}_{1+s}(A|P))}{ds}|_{s=1} \le R $, 
the maximum $\max_{0\le s }  s (H_{1+s}(A|P^{A}) - R) $
is realized at $s \in [0,1]$, i.e., 
$\max_{0\le s \le 1}  s (H_{1+s}(A|P^{A}) - R) = \max_{0\le s }  s (H_{1+s}(A|P^{A}) - R) $.
Therefore, we obtain the converse part (\ref{1-5-1}).

In order to show the direct part (\ref{1-5-2}), 
we will show the following lemma
by employing Lemma \ref{LemDI}.
\begin{lem}\Label{Lem9-9-2}
\begin{align}
& \liminf_{n\to\infty} \frac{-1}{n}\log \min_{f_n\in {\cal F}_n(R) } d_1(P^{f_n(A_n)}) \nonumber \\
\ge & \max_{0\le s \le 1}  s (H_{1+s}(A|P^{A}) - R) 
\Label{9-7-10-b}.
\end{align}
\end{lem}

In order to show Lemma \ref{Lem9-9-2},
we prepare the following lemma,
whose proof is given in Appendix \ref{al9-9-1}.

\begin{lem}\Label{l9-9-1}
When 
$\frac{d (s{H}_{1+s}(A|P))}{ds}|_{s=1} \le R $,
\begin{align}
&\min_{Q:H(Q)+D(Q\|P) \ge R}
H(Q)+2D(Q\|P) -R \nonumber \\
=&
\max_{0\le s }  s {H}_{1+s}(A|P) -s R \nonumber \\
=&
\max_{0\le s \le 1}  s {H}_{1+s}(A|P) -s R .
\Label{9-9-4}
\end{align}
When 
$\frac{d (s{H}_{1+s}(A|P))}{ds}|_{s=1} 
%\tilde{H}_{2}'(A|P) 
> R $,
\begin{align}
&\min_{Q:H(Q)+D(Q\|P) \ge R}
H(Q)+2D(Q\|P) -R \nonumber \\
=&
{H}_{2}(A|P) - R \Label{9-9-5} \\
=&
\max_{0\le s \le 1}  s{H}_{1+s}(A|P) -s R .
\Label{2-2-8}
\end{align}
\end{lem}

\quad {\it Proof of Lemma \ref{Lem9-9-2}:}~
Due to (\ref{9-6-1}), (\ref{9-7-5}),
and the continuity of $Q \mapsto H(Q)$ and $ D(Q\|P^A)$,
we obtain
\begin{align}
& \lim_{n \to \infty} \frac{-1}{n}\log
2 |{\cal T}_n| \max_{Q\in {\cal T}_n^2[\lfloor e^{nR} \rfloor ]} e^{-n (D(Q\|P^A) + H(Q))} \nonumber \\
= & 
\lim_{n \to \infty} 
\min_{Q\in {\cal T}_n^2[\lfloor e^{nR} \rfloor ]} D(Q\|P^A) + H(Q) \nonumber \\
= &
\min_{Q: D(Q\|P^A) \ge R} D(Q\|P^A) + H(Q) \nonumber \\
\ge &
\min_{Q: D(Q\| P^A) \ge R} H(Q)+2D(Q\| P^A) -R \nonumber \\
\ge & 
\min_{Q:H(Q)+D(Q\| P^A) \ge R} H(Q)+2D(Q\| P^A) -R .
\Label{9-7-12}
\end{align}
From (\ref{9-7-14}),
\begin{align*}
K_n := \sum_{Q\in {\cal T}_n^1[\lfloor e^{nR} \rfloor]} \lfloor e^{nR} \rfloor (P^A)^n(T_n(Q)) e^{-n (D(Q\|P^A) + H(Q)) }
\end{align*}
satisfies that
\begin{align*}
& \max_{Q\in {\cal T}_n^1[\lfloor e^{nR} \rfloor]} 
\frac{1}{{\cal T}_n}e^{-n (2 D(Q\|P^A) + H(Q)-R)} \\
\le &
K_n \le 
{\cal T}_n
\max_{Q\in {\cal T}_n^1[\lfloor e^{nR} \rfloor]} 
e^{-n (2 D(Q\|P^A) + H(Q)-R)}.
\end{align*}
Due to (\ref{9-6-1}) and 
the continuity of $Q \mapsto H(Q)$ and $ D(Q\|P^A)$,
\begin{align}
&\lim_{n\to\infty} \frac{-1}{n}\log K_n \nonumber \\
=&
\min_{Q:H(Q)+D(Q\| P^A) \ge R}
H(Q)+2D(Q\| P^A) -R .
\Label{9-7-9}
\end{align}
As is shown in Lemma \ref{l9-9-1}, RHSs of (\ref{9-7-12}) and (\ref{9-7-9})
equal
$\max_{0\le s \le 1}  s {H}_{1+s}(A|P^{A}) -s R$.
%(Both cases in Lemma \ref{l9-9-1} are required for this derivation.)
Since $\max_{0\le s}  s {H}_{1+s}(A|P^{A}) -s R \ge
\max_{0\le s \le 1}  s {H}_{1+s}(A|P^{A}) -s R$,
(\ref{9-7-11}) implies that
\begin{align}
&\lim_{n \to \infty} \frac{-1}{n}\log 
(P^A)^n \{ a\in{\cal A}^n| (P^A)^n(a)\ge \frac{2}{e^{nR}}\} \nonumber \\
\ge &
\max_{0\le s \le 1}  s ({H}_{1+s}(A|P^{A}) - R ).
\Label{9-7-11-f}
\end{align}
Thus, applying (\ref{9-7-12}), (\ref{9-7-9}), and (\ref{9-7-11-f}) to the RHS of (\ref{9-9-1}),
and using Lemma \ref{l9-9-1},
we can choose 
a sequence $\{f_n\}$ such that
\begin{align}
& \liminf_{n\to\infty} \frac{-1}{n}
\log \min_{f_n} d_1(P^{f_n(A_n)}) \nonumber \\
\ge & \max_{0\le s \le 1}  s (H_{1+s}(A|P^{A}) - R) ,
\end{align}
which implies (\ref{9-7-10-b}).
\endproof

\section{Secret key generation without communication}\Label{s42}
\subsection{Application of Theorem \ref{thm1}}\Label{s421}
Next, we consider the secure key generation problem from
a common random number $A \in \cA$ which has been partially eavesdropped on by Eve.
For this problem, it is assumed that Alice and Bob share a common random number $A \in \cA$,
and Eve has another random number $E \in \cE$, which is correlated to the random number $A$. 
The task is to extract a common random number 
$f(A)$ from the random number $A \in \cA$, which is almost independent of 
Eve's random number $E \in \cE$.
Here, Alice and Bob are only allowed to apply the same function $f$ to the common random number $A \in \cA$.

Then, 
when the initial random variables $A$ and $E$ obey the distribution $P^{A,E}$,
Eve's distinguishability 
can be represented by the following value:
\begin{align*}
d_1(P^{f(A),E}|E):=
d_1(P^{f(A),E},P^{f(A)}_{\mix}\times P^{E}) ,
\end{align*}
where $P^{f(A)}_{\mix}\times P^{E}$
is the product distribution of both marginal distributions $P^{f(A)}_{\mix}$ and $P^{E}$, 
and $P^{f(A)}_{\mix}$ is the uniform distribution on $\{1, \ldots, M\}$.
While the half of this value directly gives the probability that Eve can distinguish Alice's information,
we call $d_1(P^{f(A),E}|E)$ Eve's distinguishability in the following.
This criterion was proposed by \cite{Can} and was used by \cite{Ren05}.
Since the half of this quantity $d_1(P^{f(A),E}|E)$
is closely related to universally composable security,
we adopt it as the secrecy criterion in this paper.
As another criterion, we sometimes treat
\begin{align*}
d_1'(P^{f(A),E}|E):=
d_1(P^{f(A),E},P^{f(A)}\times P^{E}) .
\end{align*}
Since  
$d_1(P^{f(A)}\times P^{E},P^{M}_{\mix}\times P^{E}) 
=d_1(P^{f(A)},P^{M}_{\mix}) \le d_1(P^{f(A),E},P^{M}_{\mix}\times P^{E})$,
we have
\begin{align*}
d_1'(P^{f(A),E}|E)\le 2 
d_1(P^{f(A),E}|E).
\end{align*}
Further, when $P^{f(A)}$ is the uniform distribution,
the above criteria coincide with each other.

Next, we consider an ensemble of universal$_2$ hash functions 
$\{f_{\bX}\}$.
Similar to (\ref{1-5-4}), the equation
\begin{align}
\rE_{\bX}
d_1(P^{f_{\bX}(A),E}|E)
=
d_1(P^{B,E,\bX}, P^{B}_{\mix}\times P^E \times P^{\bX}) 
\end{align}
holds,
where $B$ is the random variable $f_{\bf X}(A)$.
Hence, when the expectation $\rE_{\bX} d_1(P^{f_{\bf X}(A),E}|E)$
is sufficiently small,
the random variable $f_{\bf X}(A)$ 
is almost independent of the random variables $\bX$ and $E$.
So, the above value is suitable even when we randomly choose 
a hash function.

In order to evaluate the average performance,
we define the quantity
\begin{align*}
\phi(t|A|E|P^{A,E})
&:= \log 
\sum_e P^E(e) 
(\sum_a P^{A|E}(a|e)^{\frac{1}{1-t}})^{1-t} \\
&= \log 
\sum_{e} 
(\sum_a P^{A,E}(a,e)^{\frac{1}{1-t}})^{1-t} .
\end{align*}
Note that when Eve's random variable $E$ takes a continuous value in the set ${\cal E}$,
the relation (\ref{4-25-1}) holds by defining 
$\phi(t|A|E|P^{A,E})$ in the following way.
\begin{align*}
\phi(t|A|E|P^{A,E}):= \log 
\int_{{\cal E}} P^E(e) de  
(\sum_a P^{A|E}(a|e)^{\frac{1}{1-t}})^{1-t}.
\end{align*}
This definition does not depend on the choice of the measure on ${\cal E}$.
%That is, when $\tilde{P}^E(e) f(e)=P^E(e)$ for a positive function $f$,
%\begin{align*}
%\phi(t|A|E|P^{A,E}):= \log \int_{{\cal E}} \tilde{P}^E(e) f(e) de
%(\sum_a P^{A|E}(a|e)^{\frac{1}{1-t}})^{1-t}.
%\end{align*}

By using Theorem \ref{thm1}
and putting $t=\frac{s}{1+s}$, 
any universal$_2$ hash functions $\{f_{\bX}\}$
satisfies
the inequality:
\begin{align}
\rE_{\bX}
d_1(P^{f_{\bX}(A),E}|E)
&\le 3 M^\frac{s}{1+s} \rE_{e} (\sum_a P^{A|E}(a|e)^{1+s})^{\frac{1}{1+s}} \nonumber \\
&= 3 M^t e^{\phi(t|A|E|P^{A,E})}\Label{4-25-1}
\end{align}
for $0 \le t \le \frac{1}{2}$.
Therefore, there exists a function $f$ such that
\begin{align}
d_1(P^{f(A),E}|E)
&\le 3 M^\frac{s}{1+s} \rE_{e} (\sum_a P^{A|E}(a|e)^{1+s})^{\frac{1}{1+s}} \nonumber \\
&= 3 M^t e^{\phi(t|P^{A,E})}\Label{4-25-1-b}.
\end{align}

Next, we consider the case when our distribution $P^{A_n E_n}$ 
is given by the $n$-fold independent and identical distribution of 
$P^{AE}$, i.e, $(P^{A,E})^n$.
Ahlswede and Csisz\'{a}r \cite{AC93} showed that
the optimal generation rate
\begin{align*}
& G(P^{A,E}) \\
:=&
\sup_{\{(f_n,M_n)\}}
\left\{
\lim_{n\to\infty} \frac{\log M_n}{n}
\left|
\lim_{n\to\infty} d_1(P^{f_{n}(A_n),E_n}|E_n)=0 
\right. \right\}
\end{align*}
equals the conditional entropy $H(A|E)$.
That is, any achievable generation rate $R= \lim_{n\to\infty} \frac{\log M_n}{n}$
is no more than $H(A|E)$.
The quantity $d_1(P^{f_{n}(A_n),E_n}|E_n)$ goes to zero.
In order to treat the speed of this convergence,
we focus on the supremum of  
the {\it exponential rate of decrease (exponent)} for 
$d_1(P^{f_{n}(A_n),E_n}|E_n)$ for a given $R$
\begin{align*}
&e_1(P^{A,E}|R) \\
:=&
\sup_{\{(f_n,M_n)\}}
\Bigl\{
\lim_{n\to\infty} 
\frac{-1}{n}\log d_1(P^{f_{n}(A_n),E_n}|E_n)
\Bigr| \\
&\hspace{27ex} %\begin{array}{l}
%\displaystyle
\lim_{n\to\infty} \frac{-1}{n} \log M_n
\le R
%\end{array}
\Bigr\}.
\end{align*}
Since the relation
$\phi(t|A^n|E^n|(P^{A,E})^n)= n \phi(t|A|E|P^{A,E})$ holds,
the inequality (\ref{4-25-1-b}) implies that
\begin{align}
e_1(P^{A,E}|R) 
\ge 
-\phi(t|A|E|P^{A,E})-tR.
\end{align}
for $t \in [0,1/2]$.
That is, taking the maximum concerning $t \in [0,1/2]$,
we obtain 
\begin{align}
e_1(P^{A,E}|R) \ge 
e_\phi(A|E|P^{A,E}|R) ,
\Label{4-16-4}
\end{align}
where
\begin{align*}
e_\phi(A|E|P^{A,E}|R) 
&:= \max_{0 \le t \le \frac{1}{2}} 
-\phi(t|A|E|P^{A,E})-tR \\
& = \max_{0 \le s \le 1} 
-\phi(\frac{s}{1+s}|A|E|P^{A,E})-\frac{s}{1+s} R .
\end{align*}
Since 
$\left.\frac{d}{dt}\phi(t|P^{A,E})\right|_{t=0}=
\left.\frac{d (s {H}_{1+s}(A|E|P^{A,E})}{ds}  \right|_{s=0}=
-H(A|E)$,
the right hand sides of (\ref{4-16-4}) and (\ref{2-3-1})
are strictly greater than $1$ for $R<H(A|E)$.

\subsection{Comparison with the previous paper \cite{Hayashi2}}\Label{s42s}
Next, we show how better our bound is than that by the previous paper \cite{Hayashi2}.
The previous paper \cite{Hayashi2} shows the following 
in Section IIA:
there exists a sequence of functions $f_n:
{\cal A}^n \to \{1, \ldots, \lfloor e^{nR} \rfloor \}$
such that
\begin{align*}
& \lim_{n \to \infty}
\frac{-1}{n}\log 
D (P^{f_n (A_n),E_n} \| P^{f_n(A_n)}_{\mix}\times P^{E_n }) \\
\ge &
\max_{0 \le s \le 1}
s{H}_{1+s}(A|E|P^{A,E})-s R,
\end{align*}
where 
we define the function
\begin{align*}
s {H}_{1+s}(A|E|P^{A,E})
&:=-\log \sum_{a,e} P^E(e)P^{A|E}(a|e)^{1+s}  \\
&=-\log \sum_{a,e} P^{A,E}(a,e)^{1+s} P^E(e)^{-s}
\end{align*}
for $s \in [0,1]$.
Hence, applying Pinsker's inequality (\ref{1-5-3}),
we obtain 
\begin{align}
e_1(P^{A,E}|R) 
\ge & \lim_{n \to \infty}
\frac{-1}{n}\log d_1(P^{f_{n}(A_n),E_n}|E_n) \nonumber  \\
\ge &
\tilde{e}_H(A|E|P^{A,E}|R)
\Label{2-3-1}
\end{align}
where
\begin{align*}
\tilde{e}_H(A|E|P^{A,E}|R)
:=&
\max_{0 \le s \le 1}
\frac{s {H}_{1+s}(A|E|P^{A,E})-s R}{2} \\
=& \max_{0 \le t \le \frac{1}{2}}
\frac{t{H}_{\frac{1}{1-t}}(A|E|P^{A,E}) -t R}{2-2t}
\end{align*}
with $s=\frac{t}{1-t}$.
Concerning the comparison of both bounds,
we prepare the following lemma.

\begin{lem}\Label{L1-6-1}
The inequality
\begin{align}
-\frac{s}{1+s}H_{1+s}(A|E|P^{A,E})
\ge
\phi(\frac{s}{1+s}|A|E|P^{A,E})
\end{align}
holds for $s\in (0,\infty)$.
Equality holds if and only if
the R\'{e}nyi entropy $H_{1+s}(A|P^{A|E=e})$
does not depend on the choice $e$ at the support of $P^E$.
\end{lem}
\begin{proof}
Applying Jensen's inequality to the concave function $x \mapsto x^{\frac{1}{1+s}} $,
we have
\begin{align*}
& e^{-\frac{s {H}_{1+s}(A|E|P^{A,E})}{1+s}}
= 
(\sum_{e} P^E(e) \sum_{a}P^{A|E}(a|e)^{1+s})^{\frac{1}{1+s}} \\
\ge &
\sum_{e} P^E(e) (\sum_{a}P^{A|E}(a|e)^{1+s})^{\frac{1}{1+s}} 
=  
e^{\phi(\frac{s}{1+s}|A|E|P^{A,E})}.
\end{align*}
Thus, the equality condition is 
that the value $\sum_{a}P^{A|E}(a|e)^{1+s}$
does not depend on the choice $e$ at the support of $P^E$.
Hence, we obtain the desired argument.
\end{proof}

In order to compare the two bounds
$e_\phi(A|E|P^{A,E}|R)$
and $\tilde{e}_H(A|E|P^{A,E}|R)$,
we introduce the following value:
\begin{align*}
e_H(A|E|P^{A,E}|R)
:=& \max_{0 \le s \le 1}
\frac{s {H}_{1+s}(A|E|P^{A,E})-s R}{1+s} \\
=& \max_{0 \le t \le \frac{1}{2}}
t{H}_{\frac{1}{1-t}}(A|E|P^{A,E}) -t R
\end{align*}

Then, we obtain the following lemma.
\begin{lem}\Label{L1-6-2}
\begin{align}
e_\phi(A|E|P^{A,E}|R)
\ge
e_H(A|E|P^{A,E}|R)
\ge
\tilde{e}_H(A|E|P^{A,E}|R)
\Label{2-4-2}
\end{align}
for $R<H(A|E)$.
Equality in the first inequality holds if and only if
the R\'{e}nyi entropy $H_{1+s_0}(A|P^{A|E=e})$
does not depend on the choice $e$ at the support of $P^E$
for 
$s_0:=
\argmax_{0 \le s \le 1} 
-\phi(\frac{s}{1+s}|A|E|P^{A,E})-\frac{s}{1+s} R$.
Equality in the second inequality holds if and only if
$\frac{s {H}_{2}(A|E|P^{A,E})- R}{2} 
= \max_{0 \le s \le 1}
\frac{s {H}_{1+s}(A|E|P^{A,E})-s R}{1+s} $.
\end{lem}
Therefore, 
our exponent $e_\phi(A|E|P^{A,E}|R)$ is strictly better than
the exponent $\tilde{e}_H(A|E|P^{A,E}|R)$ 
by \cite[Section IIA]{Hayashi2}
except for the case satisfying the following two conditions:
(i)
$-\phi(\frac{1}{2}|A|E|P^{A,E})-\frac{1}{2} R 
= \max_{0 \le s \le 1} 
-\phi(\frac{s}{1+s}|A|E|P^{A,E})-\frac{s}{1+s} R $.
(ii)
$H_{2}(A|P^{A|E=e})$
does not depend on the choice $e$ at the support of $P^E$.

For example, we consider the following case:
${\cal A}$ equals ${\cal E}$,
the set ${\cal A}$ has a module structure,
(i.e., ${\cal A}$ is an Abelian group)
and 
the conditional distribution $P^{A|E}(a|e)$ has the form $P^A(a-e)$.
Then, 
the equality condition for the first inequality holds.
Since
\begin{align*}
&e^{\phi(\frac{s}{1+s}|A|E|P^{A,E})}
=
\sum_{e} P^E(e) (\sum_{a}P^{A|E}(a|e)^{1+s})^{\frac{1}{1+s}} \\
=&
\sum_{e} P^E(e) 
e^{-\frac{s {H}_{1+s}(A|P^{A})}{1+s}}
=
e^{-\frac{s {H}_{1+s}(A|P^{A})}{1+s}}.
\end{align*}
and
\begin{align*}
&e^{-sH_{1+s}(A|E|P^{A,E})}
=
\sum_{e} P^E(e) \sum_{a}P^{A|E}(a|e)^{1+s} \\
= &
\sum_{e} P^E(e) \sum_{a}P^{A}(a-e)^{1+s} \\
=&
\sum_{e} P^E(e) 
e^{-\frac{s {H}_{1+s}(A|P^{A})}{1+s}}
=
e^{-\frac{s {H}_{1+s}(A|P^{A})}{1+s}},
\end{align*}
bounds $e_\phi(A|E|P^{A,E}|R)$ 
and $\tilde{e}_H(A|E|P^{A,E}|R)$ 
can be simplified to 
\begin{align*}
e_{\phi}(A|E|P^{A,E}|R)
&=
e_{H}(A|E|P^{A,E}|R)
=
e_{H}(A|P^{A}|R) \\
\tilde{e}_{H}(A|E|P^{A,E}|R)
&=
\tilde{e}_{H}(A|P^{A}|R),
\end{align*}
where
\begin{align*}
e_{H}(A|P^{A}|R)
:=&
\max_{0 \le s \le 1}
\frac{s {H}_{1+s}(A|P^{A})-s R}{1+s}  \\
=& \max_{0 \le t \le 1/2} 
t {H}_{\frac{1}{1-t}}(A|P^{A})-t R \nonumber \\
\tilde{e}_{H}(A|P^{A}|R)
:=&
\max_{0 \le s \le 1}
\frac{s {H}_{1+s}(A|P^{A})-s R}{2} \nonumber \\
=& \max_{0 \le t \le 1/2} 
\frac{ t {H}_{\frac{1}{1-t}}(A|P^{A})-t R }{2-2t}.
\end{align*}
In particular, 
both exponents are numerically plotted in Fig. \ref{f3}
when ${\cal A}=\{0,1\}$, and $P^A(0)=a$, $P^A(1)=1-a$.

\begin{proof}
The first inequality and its equality condition follow 
from Lemma \ref{L1-6-1} and 
the definitions of $e_\phi(P^{A,E}|R)$ and $e_H(P^{A,E}|R)$.
The second inequality follows from
the inequality 
$\frac{1}{2} \le \frac{1}{1+s} $ for $s \in [0,1]$.
Since the equality holds only when $s=1$,
we obtain the equality condition for the second inequality.
\end{proof}

\begin{figure}[htbp]
\begin{center}
\scalebox{1.0}{\includegraphics[scale=0.8]{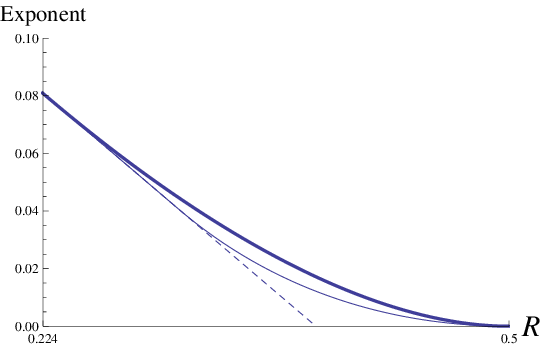}}
\end{center}
\caption{
Lower bounds of
$e_1(P^{AE}|R)$.
Thick line: 
$e_H(A|P^A|R)
%\max_{0 \le s \le 1} \frac{s {H}_{1+s}(A|P^A)- s R}{1+s}
$ (The present paper),
Normal line: $\tilde{e}_H(A|P^A|R)
%\max_{0 \le s \le 1} \frac{s {H}_{1+s}(A|P^A)- s R}{2}
$ by \cite{Hayashi2}),
Dashed line: $\frac{H_{2}(A|P^A)- s R}{2}$ (direct application of (\ref{5-14-6}) without smoothing).
Here,
$P^A$ is chosen to be the binary distribution 
$P^A(0)=\alpha$, $P^A(1)=1-\alpha$ with $\alpha=0.200$.
Then
$h(\alpha)=H(A)=0.500$, 
and $2\frac{d(s H_{1+s}(A))}{ds}|_{s=1}-H_2(A)=0.224$.}
\Label{f3}
\end{figure}%

\section{The wire-tap channel in a general framework}\Label{s2}
Next, we consider the wire-tap channel model, in which the eavesdropper (wire-tapper) Eve and the authorized receiver Bob receive the information from the authorized sender Alice.
In this case, in order for Eve to have less information, Alice chooses a suitable encoding.
This problem is formulated as follows.
Let $\cX$, $\cY$ and $\cZ$ be the alphabets of Alice, Bob, and Eve.
Then, the main channel from Alice to Bob is described by $W^B:x \mapsto W^B_x$,
and the wire-tapper channel from Alice to Eve is described by $W^E:x \mapsto W^E_x$.
That is, 
$W^B_x$ is the output distribution on Bob's side with Alice's input $x$,
and $W^E_x$ is the output distribution on Eve's side with Alice's input $x$.
In this setting, 
in order to send a secret message in $\{1, \ldots, M\}$
subject to the uniform distribution,
Alice chooses $M$ 
distributions $Q_1, \ldots, Q_M$ on $\cX$,
and she generates $x\in \cX$ subject to $Q_i$
when she wants to send the message $i \in \{1, \ldots, M\}$.
Bob prepares $M$ disjoint subsets
$\cD_1,\ldots, \cD_M$ of $\cY$ and 
judges that a message is $i$ if $y$ belongs to $\cD_i$.
Therefore, the triplet $(M,\{Q_1, \ldots, Q_M\},
\{\cD_1,\ldots, \cD_M\})$ is called a
code, and is described by $\Phi$.
Its performance is given by the following three quantities.
The first is the size $M$, which is denoted by $|\Phi|$.
The second is the average 
error probability $\epsilon_B(\Phi)$:
\begin{align*}
\epsilon_B(\Phi)\defeq
\frac{1}{M} \sum_{i=1}^M  W_{Q_i}^B (\cD_i^c),
\end{align*}
and the third is Eve's distinguishability
$d_1(\Phi|E)$:
\begin{align*} 
d_1(\Phi|E):= &
d_1(
W^E_{\Phi}\times P^M_{\mix}, W^E[\Phi]) \\
W^E_{\Phi}(e)  :=& \sum_i \frac{1}{M} W_{Q_i}^E(e)
,\quad
W^E[\Phi](i,e)  := \frac{1}{M} W_{Q_i}^E(e).
\end{align*}
The quantity $d_1(\Phi|E)$
gives an upper bound for the probability that 
Eve can succeed in distinguishing whether Alice's information belongs to 
a given subset.
So, the value can be regarded as Eve's distinguishability.
In order to calculate these values, 
we introduce the following quantity.
\begin{align*}
\phi(t|W,p) &:= \log 
\sum_y \left( \sum_x p(x)
(W_x(y))^{1/(1-t)}\right)^{1-t} .
\end{align*}
When the random variable $Y$ takes a continuous value in the set ${\cal Y}$
while $X$ takes discrete value,
the above definition can be changed to 
\begin{align*}
\phi(t|W,p) &:= \log 
\int_{{\cal Y}} \left( \sum_x p(x)
(W_x(y))^{1/(1-t)}\right)^{1-t} dy .
\end{align*}
This definition does not depend on the choice of the measure on ${\cal Y}$.
That is, when $\tilde{W}_x(y)f(y)=W_x(y)$ for a positive function $f$,
\begin{align*}
\phi(t|W,p) &= \log 
\int_{{\cal Y}} \left( \sum_x p(x)
(\tilde{W}_x(y))^{1/(1-t)}\right)^{1-t} f(y) dy .
\end{align*}

%The following lemma gives the property of $\phi(t|W,p)$.
As is shown as Lemma 1 of \cite{Hayashi2},
$\phi(t|W,p)$ satisfies the following lemma.
\begin{lem}\Label{l1}
The function $p \mapsto 
e^{\phi(t|W,p)}$ 
is convex for $t\in [-1,0]$, and is concave for $t\in [0,1]$.
\end{lem}

Now, using the function $\phi(t)$,
we make a code for the wire-tap channel based on the random coding method.
For this purpose, we make a protocol to share a random number.
First, we generate the random code $\Phi(\bY)$ with size $LM$,
which is described as $\Phi(\bY)(a)=Y_a$ for $a=1, \ldots, LM$
by using the $LM$ independent and identical random
variables $\bY=(Y_1,\ldots,Y_{ML})$ subject to the distribution $p$ on $\cX$.
Gallager \cite{Gal} 
showed that the ensemble expectation of the average error probability 
concerning decoding the input message $A$ is less than
$(ML)^{t}e^{\phi(-t|W^B,p)}$ for $0 \le t \le 1$
when Bob applies the maximum likelihood decoder $\cD'(\bY)$ of the code $\Phi(\bY)$.
After sending the random variable $A$ taking values in the set with the cardinality $ML$, 
Alice and Bob apply the above
universal$_2$ hash functions $f_{\bX}$ to the random variable $A$ and generate another piece of data of size $M$.
Here, we assume that the ensemble $\{ f_{\bX} \}$ satisfies 
Condition \ref{C12}.
Then, Alice and Bob 
share the random variable $f_{\bX}(A)$ with size $M$.
This protocol is denoted by $\Phi(\bX,\bY)'$. 

Let $E$ be the random variable of the output of Eve's channel $W^E$. 
When $p$ is the uniform distribution on the set ${\cal C}:=\{1, \ldots, ML\}$
and the joint distribution $P^{C,E}$ is given by 
$P^{C,E}(c,e):=p(c) W_c^E(e)$,
the equations
\begin{align}
e^{\phi(t|P^{C,E})}
= & \frac{1}{M^t L^t}
\sum_e \left( \sum_a p(c) (W_c^E(e))^{\frac{1}{1-t}}\right)^{1-t} 
\nonumber \\
= & \frac{e^{\phi(t|W^E,p)}}{M^t L^t}
\Label{12-26-1}.
\end{align}
hold.

For a given code $\Phi(\bY)$,
we apply the inequality (\ref{4-25-1}) to Eve's distinguishability.
Then,
\begin{align}
\rE_{\bX|\bY} d_1(\Phi(\bX,\bY)'|E) 
 \le 
3 \frac{e^{\phi(t|W^E,p_{\mix,\Phi(\bY)})}}{L^t}
\Label{8-5-1}
\end{align}
for $0 \le t \le \frac{1}{2}$.
The concavity of $e^{\phi(t|W^E,p)}$ (Lemma \ref{l1}) guarantees that
\begin{align*}
\rE_{\bX,\bY} d_1(\Phi(\bX,\bY)'|E) 
 \le &
3 \rE_{\bY}\frac{e^{\phi(t|W^E,p_{\mix,\Phi(\bY)})}}{L^t} \\
 \le &
3 \frac{e^{\phi(t|W^E,p)}}{L^t}
\end{align*}
for $0 \le t \le \frac{1}{2}$.

Now, we make a code for the wire-tap channel by modifying the above protocol $\Phi(\bX,\bY)'$.
First, we choose the distribution $Q_i$ to be the uniform distribution on $f_{\bX}^{-1}\{i\}$.
When Alice wants to send the secret message $i$,
before sending the random variable $A$, 
Alice generates the random number $A$ subject to the distribution $Q_i$.
Alice sends the random variable $A$.
Bob recovers the random variable $A$
by using the maximum likelihood decoder $\cD'(\bY)$, 
and applies the function $f_{\bX}$.
Then, Bob decodes Alice's message $i$, and this code for wire-tap channel $W^B,W^E$
is denoted by $\Phi(\bX,\bY)$.
Since the ensemble $\{ f_{\bX} \}$ satisfies Condition \ref{C12}
and the secret message $i$ obeys the uniform distribution
on $\{1,\ldots, M\}$,
this protocol $\Phi(\bX,\bY)$ has the same performance as the above protocol $\Phi(\bX,\bY)'$.

Finally, 
we consider what code is derived from the above random coding discussion.
Using the Markov inequality, we obtain
\begin{align*}
\rP_{\bX,\bY} 
\{ \epsilon_B(\Phi(\bX,\bY)) \le 3 \rE_{\bX,\bY} \epsilon_B(\Phi(\bX,\bY)) \}
&\ge \frac{2}{3} \\
\rP_{\bX,\bY} 
\{ 
d_1(\Phi(\bX,\bY)|E)
\le 3 \rE_{\bX,\bY} d_1(\Phi(\bX,\bY)|E)
\}
&\ge \frac{2}{3} .
\end{align*}
Therefore, the existence of a good code is guaranteed in the following way.
That is, we give the concrete performance of a code 
whose existence is shown in the above random coding method.

\begin{thm}\Label{3-6}
There exists a code $\Phi$ for any
integers $L,M$,
and any probability distribution $p$ on $\cX$
such that $|\Phi| =M $ and
\begin{align*}
\epsilon_B(\Phi) 
\le & 3 \min_{0\le t\le 1}(ML)^{t}e^{\phi(-t|W^B,p)},
\\
d_1(\Phi|E)  
\le & 9
%\min_{0 > t \ge -1/2}\frac{L^t e^{\phi(t|W^E,p)}}{-t},
\min_{0 \le t \le \frac{1}{2}}
\frac{e^{\phi(t|W^E,p)}}{L^t }.
\end{align*}
\end{thm}

In the $n$-fold discrete memoryless channels $W^{B_n}$ and $W^{E_n}$ 
of the channels $W^B$ and $W^E$,
the additive equation
$\phi(t|W^{B_n},p)= n \phi(t|W^B,p)$ 
holds.
Thus, there exists a code $\Phi_n$ for any
integers $L_n,M_n$,
and any probability distribution $p$ on $\cX$
such that $|\Phi_n| =M_n $ and
\begin{align*}
\epsilon_B(\Phi)  
\le & 3
\min_{0\le t\le 1}
(M_n L_n)^{t}e^{n \phi(-t|W^B,p)} ,\\
d_1(\Phi_n|E)  
\le & 9
%\min_{0 > t \ge -1/2}\frac{L^t e^{\phi(t|W^E,p)}}{-t},
\min_{0 \le t \le \frac{1}{2}}
\frac{e^{n \phi(t|W^E,p)}}{L_n^t}.
\end{align*}
Since 
$\lim_{t \to 0} \frac{\phi(t|W^{E},p)}{t}= I(p:W^E)$,
the rate $\max_p I(p:W^B)-I(p:W^E)$ can be asymptotically 
attained.
Therefore,
when the sacrifice information rate is $R$, i.e., $L_n\cong e^{nR}$,
the exponential rate of decrease for Eve's distinguishability is 
greater than
\begin{align*}
e_{\phi}(R|W^E,p):=
\max_{0 \le t \le 1/2} t R -\phi(t|W^E,p).
\end{align*}

\section{Comparison with existing bounds}\Label{s6-1}
In Subsection \ref{s6-1-1}, 
we compare our exponent 
$e_{\phi}(R|W^E,p)$ with those derived by \cite{Hayashi,Hayashi2}
in the general setting.
In Subsections \ref{s6-1-2} and \ref{s6-1-3},
using discussion in Subsection \ref{s42s}, we treat this comparison 
in special cases more deeply.

\subsection{General case}\label{s6-1-1}
Now, we compare the lower obtained bound
$e_{\phi}(R|W^E,p)$ for 
the exponential rate of decrease for Eve's distinguishability
with existing lower bounds \cite{Hayashi,Hayashi2}.
Using the quantity
\begin{align}
\psi(t|W,p) &:= \log 
\sum_y \left( \sum_x p(x)
(W_x(y))^{1+t}\right)W_p(y)^{-t}\\
W_p(y) &:= \sum_x p(x)W_x(y),\nonumber
\end{align}
the previous paper \cite{Hayashi} derived
the following lower bound of this exponential rate of decrease:
\begin{align}
e_{\psi}(R|W^E,p)
:=&
\max_{0\le s\le1} \frac{s R-\psi(s|W^E,p)}{1+s} \nonumber \\
=&\max_{0 \le t \le 1/2} t R-(1-t)\psi(\frac{t}{1-t}|W^E,p).
\end{align}
The other previous paper \cite{Hayashi2} also derived
the following lower bound:
\begin{align}
\max_{0\le s\le1} s R-\psi(s|W^E,p)
\Label{1-5-5}
\end{align}
for the exponential rate of decrease for the mutual information.
By applying a discussion similar to Subsection \ref{s42s}
and Pinsker's inequality (\ref{1-5-4}),
the bound (\ref{1-5-5})
yields the bound
\begin{align}
\tilde{e}_{\psi}(R|W^E,p)
:=\max_{0\le s\le1} \frac{s R-\psi(s|W^E,p)}{2},
\Label{1-6-18}
\end{align}
which is smaller than 
the lower bound $e_{\psi}(R|W^E,p)$ because $\frac{1}{2} \le \frac{1}{1+s}$ for $0 \le s \le 1$.
Hence, in order to show the superiority of our bound $e_{\phi}(R|W^E,p)$,
it is sufficient to show the superiority over the bound 
$e_{\psi}(R|W^E,p)$.

In the following, we compare the two bounds 
$e_{\phi}(R|W^E,p)$ and $e_{\psi}(R|W^E,p)$.
For this purpose, we treat
$e^{\phi(t|W^E,p)}$
and $e^{(1-t)\psi(\frac{t}{1-t}|W^E,p)}$ for $0 \le t \le \frac{1}{2}$.
The reverse H\"{o}lder inequality \cite{RH} for the measurable space $({\cal X},p)$ is 
\begin{align*}
& \sum_{x\in {\cal X}} p(x) | X(x)Y(x)| \\
\ge &
(\sum_{x\in {\cal X}} p(x) |X(x)|^{\frac{1}{1+s}})^{1+s}
(\sum_{x\in {\cal X}} p(x) |Y(x)|^{-\frac{1}{s}})^{-s}
\end{align*}
for $s \ge 0$.
Using this inequality, we obtain
\begin{align*}
& \sum_y 
\left[
\sum_x p(x)
(W_x(y))^{1+s} \right]
W_p(y)^{-s} \\
\ge & 
\left(
 \sum_y 
\left[
\sum_x p(x)
(W_x(y))^{1+s} \right]^{\frac{1}{1+s}}
\right)^{1+s}
\cdot 
\left(
\sum_y 
W_p(y)^{-s\cdot -\frac{1}{s}}
\right)^{-s} \\
= &
\left(
 \sum_y 
\left[
\sum_x p(x)
(W_x(y))^{1+s} \right]^{\frac{1}{1+s}}
\right)^{1+s}.
\end{align*}
Substituting $s=\frac{t}{1-t}$,
we obtain
\begin{align*}
& \sum_y 
\left[
\sum_x p(x)
(W_x(y))^{\frac{1}{1-t}} \right]
W_p(y)^{\frac{-t}{1-t}} \\
\ge &
\left(
 \sum_y 
\left[
\sum_x p(x)
(W_x(y))^{\frac{1}{1-t}} \right]^{1-t}
\right)^{\frac{1}{1-t}},
\end{align*}
which implies
\begin{align*}
& e^{(1-t)\psi(\frac{t}{1-t}|W^E,p)} \\
= &
\left(
\sum_y 
\left[
\sum_x p(x)
(W_x(y))^{\frac{1}{1-t}} \right]
W_p(y)^{\frac{-t}{1-t}}
\right)^{1-t} \\
\ge &
\sum_y 
\left[
\sum_x p(x)
(W_x(y))^{\frac{1}{1-t}} \right]^{1-t}
=
e^{\phi(t|W^E,p)}.
\end{align*}
Thus, our bound 
$e_{\phi}(R|W^E,p)$ for the exponential rate of decrease is 
better than the existing bound 
$e_{\psi}(R|W^E,p)$ \cite{Hayashi}.

\begin{example}
Assume that ${\cal X}={\cal E}=\{0,1\}$.
We consider the following channel.
\begin{align*}
W_0(0)=\alpha, ~
W_0(1)=1-\alpha,~
W_1(0)=1-9\alpha, ~
W_1(1)=9\alpha.
\end{align*} 
When $p(0)=1/2, p(1)=1/2$,
\begin{align*}
I(p,W)= &h(1/2-5\alpha)-\frac{(h(\alpha)+h(9\alpha)}{2} \\
\psi(t|p,W)= &\log \Biggl( (\frac{\alpha^{1 + t} + (1 - 9\alpha)^{1 + t}}{2} (\frac{1}{2} - 5\alpha)^{-t}\\
   & + (\frac{(9\alpha)^{1 + t} + (1 - \alpha)^{1 + t}}{2}(1/2 + 5\alpha)^{-t} )\Biggr) \\
\phi(t|p,W)=&
 \log \Biggl( (\frac{\alpha^{1/(1 - t)} + (1 - 9\alpha)^{1/(1 - t)}}{2})^{1-t} \\
& + ( \frac{ (9\alpha)^{1/(1 - t)} + (1 - \alpha)^{1/(1 - t)}}{2})^{1 - t}\Biggr).
\end{align*} 
Then, the three bounds
$e_{\phi}(R|W,p)$, $e_{\psi}(R|W,p)$, and 
$\tilde{e}_{\psi}(R|W,p)$
%$\max_{0\le s\le1} \frac{s R-\psi(s|W,p)}{2}$
with $\alpha=0.05$
are numerically compared as in Fig. \ref{f4}.

\begin{figure}[htbp]
\begin{center}
\scalebox{0.9}{\includegraphics[scale=1]{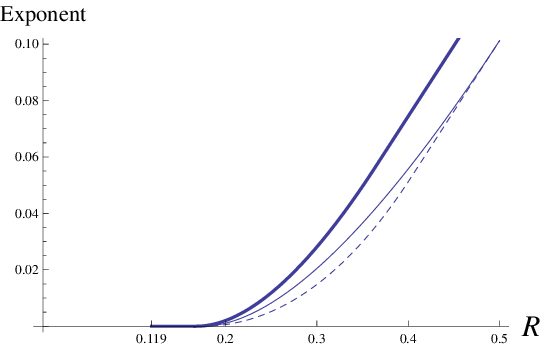}}
\end{center}
\caption{
Lower bounds of exponent.
Thick line: $e_{\phi}(R|W,p)$ (The present paper),
Normal line: $e_{\psi}(R|W,p)$ \cite{Hayashi},
Dashed line: $\tilde{e}_{\psi}(R|W,p)$ \cite{Hayashi2}.
Here, 
$\alpha$ is chosen to be $0.0500$.
Then, $I(p,W)=0.119$.
}
\Label{f4}
\end{figure}%
\end{example}

\subsection{Additive case}\label{s6-1-2}
Next, we consider a more specific case.
When ${\cal X}={\cal Z}$ and ${\cal X}$ is a module and $W_{x}(z)=W_{0}(z-x)
=P^X(z-x)$, 
the channel $W$ is called {\it additive}.

Since
\begin{align}
& e^{(1-t)\psi(\frac{t}{1-t}|W^E,p_{\mix})}
=e^{\phi(t|W^E,p_{\mix})}\nonumber \\
=& |{\cal X}|^{t} e^{- t{H}_{\frac{1}{1-t}}(X|P^X)}\Label{12-26-2},
\end{align}
any additive channel $W^E$ satisfies 
\begin{align}
& e_{\psi}(R|W^E,p_{\mix})= e_{\phi}(R|W^E,p_{\mix}) \nonumber \\
%=& \max_{0 \le t \le \frac{1}{2}}
%t (R-\log |{\cal X}|)
%+ t {H}_{\frac{1}{1-t}}(X|P) \\
=& \max_{0 \le t \le \frac{1}{2}}
t (R-\log |{\cal X}|)+t H_{\frac{1}{1-t}}(X|P^X)) \nonumber \\
=& e_H(X|P^X| \log |{\cal X}|-R) \Label{1-5-7}
\end{align}
and 
\begin{align*}
& \tilde{e}_{\psi}(R|W^E,p_{\mix}) 
= \max_{0 \le t \le \frac{1}{2}}
\frac{t(R-\log |{\cal X}|) +t {H}_{\frac{1}{1-t}}(X|P^X)}{2-2t}  \\ 
=& \tilde{e}_H(X|P^X| \log |{\cal X}|-R) 
\end{align*}
for the uniform distribution $p_{\mix}$ on ${\cal X}$.

Hence, our bound $e_{\phi}(R|W^E,p_{\mix})$ is the same as
the previous bound $e_{\psi}(R|W^E,p_{\mix})$.
However,
since $\frac{1}{2-2t} < 1$ for $t \in [0,1/2)$,
our bound 
$e_{\phi}(R|W^E,p_{\mix})$ is strictly better than
the bound $\tilde{e}_{\psi}(R|W^E,p_{\mix})$
by the other previous paper \cite{Hayashi2}
when the maximum is attained by $t \in [0,1/2)$.

\subsection{General additive case}\label{s6-1-3}
We consider a more general case.
Eve is assumed to have two random variables 
$Z \in {\cal X}$ and $Z' \in {\cal Z}'$.
The first random variable $Z$ is the output of an additive channel 
depending on the second variable $Z'$.
That is, the channel $W_x^E(z,z')$ can be written as
$W_x^E(z,z')=P^{X,Z'}(z-x,z')$, where $P^{X,Z'}$ is a joint distribution.
Hereafter, this channel model is called a general additive channel.
This channel is also called 
a regular channel \cite{DP}.
For this channel model, we obtain
\begin{align}
& e^{\phi(s|W^E,P_{\mix,\cX})}
= \sum_{z,z'} (
\sum_x \frac{1}{|{\cal X}|} W_x^E(z,z')^{\frac{1}{1-s}} )^{1-s}
\nonumber \\
= & \sum_{z,z'} 
(\sum_x \frac{1}{|{\cal X}|} P^{X,Z'}(z-x,z')^{\frac{1}{1-s}} )^{1-s}\nonumber \\
= & 
\frac{1}{|{\cal X}|^{1-s}}
\sum_{z,z'} 
(\sum_x  P^{X,Z'}(-x,z')^{\frac{1}{1-s}} )^{1-s}\nonumber \\
= & 
\frac{|{\cal X}| }{|{\cal X}|^{1-s}}
\sum_{z'}
(\sum_x P^{X,Z'}(x,z')^{\frac{1}{1-s}} )^{1-s}\nonumber \\
= & 
|{\cal X}|^s 
e^{\phi(s|X|Z'|P^{X,Z'})}, \Label{5-14-1}
\end{align}
and
\begin{align}
& 
e^{\psi(s|W^E,p_{\mix})}\nonumber \\
=& \sum_{z,z'} 
(\sum_x \frac{1}{|{\cal X}|} W_x^E(z,z')^{1+s} )
(\sum_x \frac{1}{|{\cal X}|} W_x^E(z,z') )^{-s} \nonumber \\
=& 
|{\cal X}|^{s-1}
\sum_{z,z'} 
(\sum_x P^{X,Z'}(z-x,z')^{1+s} )
(\sum_x P^{X,Z'}(z-x,z') )^{-s} \nonumber \\
=& 
|{\cal X}|^{s-1}
\sum_{z,z'} 
(\sum_x P^{X,Z'}(-x,z')^{1+s} ) P^{Z'}(z')^{-s} \nonumber \\
=& 
|{\cal X}|^{s-1}
|{\cal X}|
\sum_{z'} \sum_x P^{X,Z'}(x,z')^{1+s} P^{Z'}(z')^{-s} \nonumber \\
=& 
|{\cal X}|^{s}
e^{-s H_{1+s}(X|Z'|P^{X,Z'})}.
\end{align}
Then, the equalities
\begin{align}
& e_{\phi}(R|W^E,p_{\mix}) \nonumber \\
=& \max_{0 \le t \le \frac{1}{2}}
t (R-\log |{\cal X}|) - \phi(t|X|Z'|P^{X,Z'}) \nonumber \\
=& e_{\phi}(X|Z'|P^{X,Z'}| \log |{\cal X}|-R) ,\\
& e_{\psi}(R|W^E,p_{\mix}) \nonumber \\
=& \max_{0 \le t \le \frac{1}{2}}
t (R-\log |{\cal X}|) +t{H}_{\frac{1}{1-t}}(X|Z'|P^{X,Z'}) 
\nonumber \\
=& e_{H}(X|Z'|P^{X,Z'}| \log |{\cal X}|-R)\Label{1-6-17}, \\
%=& \max_{0 \le t \le \frac{1}{2}}
%t (R-\log |{\cal X}| +H_{\frac{1}{1-t}}(X|Z'|P^{X,Z'})) \nonumber \\
%=& \max_{0 \le s \le 1}
%\frac{s}{1+s} 
%(R-\log |{\cal X}| +H_{1+s}(X|Z'|P^{X,Z'})) \nonumber \\
& \tilde{e}_{\psi}(R|W^E,p_{\mix}) \nonumber \\
=& \max_{0 \le t \le \frac{1}{2}}
\frac{t (R-\log |{\cal X}|) 
+t{H}_{\frac{1}{1-t}}(X|Z'|P^{X,Z'}) }{2-2t}
\nonumber \\
=& \tilde{e}_{H}(X|Z'|P^{X,Z'}| \log |{\cal X}|-R)
\Label{1-6-14}
%=& \max_{0 \le s \le 1}
%\frac{s}{2} (R-\log |{\cal X}| +{H}_{1+s}(X|Z'|P^{X,Z'}) )
%\nonumber
\end{align}
hold.

Hence, 
the observation in Section \ref{s42s}
can be applied to
the comparison among 
$e_{\phi}(R|W^E,p_{\mix})$,
$e_{\psi}(R|W^E,p_{\mix})$,
and
$\tilde{e}_{\psi}(R|W^E,p_{\mix})$.
Due to Lemma \ref{L1-6-2},
$e_{\phi}(R|W^E,p_{\mix})$
is strictly better than
$e_{\psi}(R|W^E,p_{\mix})$
and
$\tilde{e}_{\psi}(R|W^E,p_{\mix})$
except for the special case mentioned in Lemma \ref{L1-6-2}.

\section{Wire-tap channel with linear coding}\Label{s3-a}
In a practical sense, 
we need to take into account the decoding time.
For this purpose, we often restrict our codes to linear codes.
In the following, we consider the case where 
the sender's space $\cX$ has the structure of a module.
When 
an error correcting code is given as a submodule $C_1\subset \cX$
and
the decoder by the authorized receiver is given as $\{\cD_x\}_{x\in C_1}$,
our code for a wire-tap channel
is given as
$\Phi_{C_1,C_2}= (|C_1/C_2|,
\{Q_{[x]}\}_{[x]\in C_1/C_2},
\{\cD_{[x]}\}_{[x]\in C_1/C_2})$
based on a submodule $C_2$ of $C_1$ as follows.
The encoding $Q_{[x]}$ is given as the uniform distribution on the coset $[x]:=x+C_2$, and 
the decoding $\cD_{[x]}$ is given as the subset 
$\cup_{x'\in x+C_2} \cD_{x'}$.
Next, we consider a submodule 
$C_2(\bX)$ of $C_1$ with cardinality $|C_2(\bX)|=L$
that is labeled by a random variable $\bX$.
Then, the module $C_2(\bX)$ can be regarded as a random variable.
Now, we impose the module $C_2(\bX)$ the following condition.

\begin{condition}\Label{C3}
Any element $x \neq 0 \in C_1$
is included in $C_2(\bX)$ with probability at most 
$\frac{L}{|C_1|}$.
\end{condition}

Then, using (\ref{8-5-1}), we can evaluate 
the performance of the constructed code in the following way.
\begin{thm}
Choose the subcode $C_2(\bX)$ according to Condition \ref{C3}.
We construct the code $\Phi_{C_1,C_2(\bX)}$ by choosing the distribution 
$Q_{[x]}$ to be the uniform distribution on $[x]$ for $[x]\in C_1/C_2(\bX)$.
Then, we obtain
\begin{align}
\rE_{\bX} d_1(\Phi_{C_1,C_2(\bX)}|E) 
\le &
3 \frac{e^{\phi(t|W^E,P_{\mix, C_1})}}{L^t}
\quad 0 \le \forall t \le \frac{1}{2},\Label{4-27-1}
\end{align}
where $P_{\mix,S}$ is the uniform distribution on the subset $S$.
\end{thm}

When the channel $W^E$ is additive, i.e., $W^E_x(z)=P^X(z-x)$,
the equation 
$\phi(t|W^E,P_{\mix, C_1+x})= \phi(t|W^E,P_{\mix, C_1})$ 
holds for any $x$.
Thus, the concavity of $e^{\phi(t|W^E,p)}$ (Lemma \ref{l1})
implies that 
\begin{align}
\phi(t|W^E,P_{\mix, C_1})
\le \phi(t|W^E,P_{\mix, {\cal X}}).\Label{2-13-1}
\end{align}
Thus, combining (\ref{4-27-1}), (\ref{2-13-1}), and (\ref{12-26-2}), 
we obtain
\begin{align}
\rE_{\bX} d_1(\Phi_{C_1,C_2(\bX)}|E) 
\le &
3 \frac{|{\cal X}|^t e^{-t{H}_{\frac{1}{1-t}}(X|P)}}{L^t}
\Label{4-27-2}
\end{align}
for 
$0 < t <\frac{1}{2}$.
That is,
when $L=e^R$, taking the minimum concerning $0 < t <\frac{1}{2}$,
we obtain
\begin{align}
\rE_{\bX} d_1(\Phi_{C_1,C_2(\bX)}|E) 
\le & 3 e^{-e_{H}(X|P^X|\log |{\cal X}|-R)}.
\end{align}
When the additive noise 
obeys the $n$-fold i.i.d. of $P$
on ${\cal X}^n$
and $L=e^{nR}$,
we obtain
\begin{align}
\rE_{\bX} d_1(\Phi_{C_1,C_2(\bX)}|E) 
\le & 3 e^{-n e_{H}(X|P^X|\log |{\cal X}|-R)}.
\Label{1-6-11}
\end{align}

Similarly, when the channel $W^E$ is general additive, i.e., 
$W^E_x(z,z')=P^{X,Z'}(z-x,z')$,
combining (\ref{4-27-1}), (\ref{2-13-1}), and (\ref{5-14-1}), 
we obtain
\begin{align}
\rE_{\bX} d_1(\Phi_{C_1,C_2(\bX)}|E) 
\le &
3 \frac{|{\cal X}|^t e^{\phi(t|X|Z'|P^{X,Z'})}}
{L^t} 
\Label{4-27-3}
\end{align}
for $0 < t <\frac{1}{2}$.
That is,
when $L=e^R$, taking the minimum concerning $0 < t <\frac{1}{2}$,
we obtain
\begin{align}
\rE_{\bX} d_1(\Phi_{C_1,C_2(\bX)}|E) 
\le & 3 e^{-e_{\phi}(X|Z'|P^{X,Z'}|\log |{\cal X}|-R)}.
\end{align}
In the $n$-fold i.i.d. case,
when $L =e^{nR}$,
we obtain
\begin{align}
\rE_{\bX} d_1(\Phi_{C_1,C_2(\bX)}|E) 
\le & 3 e^{-n e_{\phi}(X|Z'|P^{X,Z'}|\log |{\cal X}|-R)}.
\Label{1-6-16}
\end{align}

%In the following, this method choosing $\Phi_{C_1,C_2(\bX)}$ is called universal random linear privacy amplification.
When
$\cX$ is an $n$-dimensional vector space $\bF_q^n$
over the finite field $\bF_q$,
the bound can be attained by the combination of
linear code and the concatenation of a Toeplitz matrix and the identity
$(\bX,I)$ of the size $m \times (m-k)$ \cite{Hayashi2}.
Hence, if the error correcting code $C_1$ can be realizable, 
the whole process in the above code can be realizable.

\begin{rem}\Label{rem1}
In the additive case,
due to (\ref{1-5-7}), 
the exponent of 
the upper bound given in (\ref{1-6-11}) is the same as 
that given by the previous paper \cite{Hayashi}.
However, the code given in \cite{Hayashi}
is constructed by completely random coding.
However, the code given in this section 
is based on the ordinary linear code.
For security, it requires only the universal hash condition.
So, our construction requires smaller complexity than that given in 
 \cite{Hayashi}.
In the general additive case,
our exponents (\ref{1-6-16})
is strictly better than 
that given in \cite{Hayashi},
which is calculated in (\ref{1-6-17}).

Next, we consider the relation with the other previous paper \cite{Hayashi2}
in the general additive case.
The protocol given in \cite{Hayashi2} is 
is quite similar to ours.
However, as is shown in Lemma \ref{L1-6-2},
except for the very special case,
our exponent (\ref{1-6-16}) is strictly better than 
that given in \cite{Hayashi2},
which is calculated in (\ref{1-6-14}).
% when the maximum is attained with $t=1/2$.
Remember that 
the exponent given in 
 \cite{Hayashi2} is 
$\tilde{e}_{\psi}(R|W^E,p_{\mix}) $,
which is mentioned around (\ref{1-6-18}).
\end{rem}

\section{Secret key generation with public communication}\Label{s6}
Furthermore, the above result can be applied to secret key generation (distillation) with one-way public communication,
in which,
Alice, Bob, and Eve are assumed to have initial random variables 
$A\in {\cal A}$, $B\in {\cal B}$, and $E\in {\cal E}$, respectively.
The task for Alice and Bob is to share a common random variable almost independent of Eve's random variable $E$ by using a public communication.
For this purpose,
we assume that 
Alice and Bob can perform local data processing in the both sides
and Alice can send messages to Bob via public channel.
That is, only one-way communication is allowed.
We call such a combination of these operations a code and denote it by $\Phi$.

The quality is evaluated by three quantities:
the size of the final common random variable,
the probability that their final variables coincide, and
Eve's distinguishability $d_1(\Phi|E)$ of the final joint distribution between Alice and Eve.

In order to construct a protocol for this task,
we assume that the set ${\cal A}$ has a module structure 
(any finite set can be regarded as a cyclic group).
Then, the objective of secret key distillation can be realized by applying the code of a wire-tap channel
as follows.
First, 
Alice generates another uniform random variable $X$ and sends
the random variable $X':= X+A$.
Then, 
the distribution of the random variables $B,X'$ ($E,X'$) accessible to Bob (Eve)
can be regarded as the output distribution of the channel $x \mapsto W^{B}_x$
($x \mapsto W^{E}_x$).
The channels $W^{B}$ and $W^{E}$ are given as follows.
\begin{align}
W^{B}_x(x',b)= P^{A,B}(x'-x,b), ~
W^{E}_x(x',e)= P^{A,E}(x'-x,e),\Label{4-26-1}
\end{align}
where $P^{AB}(a,b)$ 
($P^{AE}(a,e)$) 
is the joint probability between
Alice's initial random variable $A$ and
Bob's (Eve's) initial random variable $B$ ($E$).
Hence, the channel $W^E$ is general additive.

Applying Theorem \ref{3-6} to the uniform distribution $P_{\mix}^A$,
for any numbers $M$ and $L$,
due to (\ref{5-14-1}),
there exists a code $\Phi$ such that $|\Phi|=M$ and
\footnote{
The previous paper \cite[Section VI]{Hayashi} derived upper bounds different from (\ref{1-6-19})
and (\ref{1-6-20}) while 
it treat the same protocol.
The previous paper \cite[Section VI]{Hayashi}
erroneously calculated 
$e^{\phi(-s|W^B,P_{\mix,\cA})}$
to 
$|\cA|^{-s} e^{-s{H}_{\frac{1}{1+s}}(A|B|P^{A,B})}$.
However, the correct calculation is 
$|\cA|^{-s} e^{\phi(-s| A|B|P^{A,B})}$
as is shown in (\ref{5-14-1}).
}
\begin{align}
\epsilon_B(\Phi)&\le 3 \min_{0\le s\le1} (ML)^s |\cA|^{-s} 
e^{\phi(-s| A|B|P^{A,B})} 
\Label{1-6-19}\\
d_1(\Phi|E)&\le 9
\min_{0\le t\le \frac{1}{2}} \frac{|{\cal A}|^t 
e^{\phi(t| A|E|P^{A,E})} }{L^t }.
\end{align}

In particular, 
when 
the joint distribution between 
$A$ and $B$($E$) 
is the $n$-fold independent and identical distribution (i.i.d.)
of $P^{A,B}$ ($P^{A,E}$), respectively,
the relation
$
\phi(t| A^n|E^n|(P^{A,E})^n)
=
n\phi(t| A|E|P^{A,E})$
hold.
Thus, there exists a code $\Phi_n$ for any
integers $L_n,M_n$,
and any probability distribution $p$ on $\cX$
such that $|\Phi_n| =M_n $ and
\begin{align}
\epsilon_B(\Phi) & \le 3
\min_{0\le s\le 1}
(M_n L_n)^{s}
|{\cal A}|^{-ns} 
e^{n\phi(-s| A|B|P^{A,B})} 
\Label{1-6-20}\\
%\Label{3-8-1-a}
d_1(\Phi_n|E) &\le 9
\min_{0\le t\le \frac{1}{2}} \frac{|{\cal A}|^{nt} 
e^{n\phi(t| A|E|P^{A,E})}}{L_n^t }.
\Label{2-13-2}
\end{align}

Finally, 
we mention the relation with the previous paper \cite{Hayashi}.
Since the above discussion is an application of 
section \ref{s3-a},
the same comparison as Remark \ref{rem1} is valid.
Hence, our evaluation (\ref{2-13-2}) is strictly better than 
that given in \cite{Hayashi}
except for the special case.

\section{Discussion}
We have derived a tight evaluation of the exponent for 
the average of the $L_1$ norm distance
between the generated random number and the uniform random
number when universal$_2$ hash functions are applied
and the key generation rate is less than the critical rate $R_1$.
Using this evaluation, 
we have obtained 
an upper bound for Eve's distinguishability in secret key generation from a common random number without communication
when universal$_2$ hash functions are applied.
Since our bound is based on the R\'{e}nyi entropy of order $1+s$ for $s \in [0,1]$,
it can be regarded as an extension of Bennett et al \cite{BBCM}'s result with the R\'{e}nyi entropy of order 2.

Applying this bound to the wire-tap channel, we obtain an 
upper bound for Eve's distinguishability,
which yields an exponential upper bound.
This exponent improves on the existing exponent \cite{Hayashi}.
Further, when the error correction code is given by a linear code
and when the channel is additive or general additive,
the privacy amplification is given by a concatenation of a Toeplitz matrix and the identity matrix.
This method can be applied to secret key distillation with public communication.

\section*{Acknowledgments}
The author is grateful to Professors Ryutaroh Matsumoto 
and Takeshi Koshiba for a helpful comments.
He is also grateful to the
referee of the previous version for helpful comments.
This research was partially supported by 
a MEXT Grant-in-Aid for Young Scientists (A) No. 20686026
and a MEXT Grant-in-Aid for Scientific Research (A) No. 23246071.
The Centre for Quantum Technologies is funded by the
Singapore Ministry of Education and the National Research Foundation
as part of the Research Centres of Excellence programme. 

\appendices

\section{Proof of Theorem \ref{thm2}}\Label{as2}
First, for a fixed element $a \in \Omega$, we introduce the condition
for a hash function $f_{\bf X}$:
\begin{condition}[{Condition $[ a, \Omega ]$}]
\begin{align*}
f_{\bf X}(a)\neq 
f_{\bf X}(a')\hbox{ for } \forall a' (\neq a) \in \Omega.
\end{align*}
\end{condition}
Let $P[a,\Omega]$ be the probability that Condition $[a,\Omega]$ holds.
Due to 
the strongly universal$_2$ condition,
it is evaluated as
\begin{align*}
1- P[a,\Omega] =&
\Pr \cup_{a' (\neq a) \in \Omega}\{ f_{\bf X}(a)= f_{\bf X}(a') \} \\
\le & \sum_{a' (\neq a) \in \Omega} \Pr \{ f_{\bf X}(a)= f_{\bf X}(a') \} \\
= &\sum_{a' (\neq a) \in \Omega} \frac{1}{M}
=\frac{|\Omega|}{M},
\end{align*}
which implies that $P[a,\Omega] \ge 1-\frac{|\Omega|}{M}$.
When we denote the expectation concerning the hash funcations under Condition $[a,\Omega]$
by $\rE_{{\bf X}| [a,\Omega] }$,
the strongly universal$_2$ condition yields that
\begin{align*}
& 
\rE_{{\bf X}| [a,\Omega] }
\Bigl(P^A(a)+
\sum_{a' (\neq a )\in f_{\bf X}^{-1}(a)} P^A(a')-\frac{1}{M}\Bigr) \\
= &
P^A(a)+
\rE_{{\bf X}| [a,\Omega] } \sum_{a' (\neq a )\in f_{\bf X}^{-1}(a)}
P^A(a')-\frac{1}{M} \\
= &
P^A(a)+\frac{1}{M} \sum_{a' (\neq a )\in \cA\setminus \Omega} P^A(a')-\frac{1}{M}\\
= &
P^A(a)+\frac{1}{M} (1-P^A(\Omega))-\frac{1}{M} \\
= &
P^A(a) -\frac{1}{M} P^A(\Omega).
\end{align*}
When Conditions $[a_1,\Omega],[a_2,\Omega], \ldots, [a_k,\Omega]$ 
hold for $a_1,a_2, \ldots, a_k \in \Omega$,
$f_{\bf X}(a_1), f_{\bf X}(a_2), \ldots, f_{\bf X}(a_k)$ are different.
Then,
\begin{align*}
d_1(P^{f_{\bf X}(A)}) 
\ge 
\sum_{j=1}^k \Bigl|P^A(a_j)+\sum_{a' (\neq a_j )\in f_{\bf X}^{-1}(a_j)} P^A(a')-\frac{1}{M}\Bigr| .
\end{align*}
Now,
we define the random variable $Y(a)$ to be $1$ when 
Condition $[a,\Omega]$ holds.
We define it to be 0 otherwise.
Then,
\begin{align*}
& d_1(P^{f_{\bf X}(A)}) \\
\ge &
\sum_{a \in \Omega} Y(a) \Bigl|P^A(a)+\sum_{a' (\neq a )\in f_{\bf X}^{-1}(a)} P^A(a')-\frac{1}{M}\Bigr| .
\end{align*}
Therefore, taking the expectation, we can evaluate $\rE_{{\bf X}}d_1(P^{f_{\bf X}(A)})$ as follows.
\begin{align*}
& \rE_{{\bf X}}
d_1(P^{f_{\bf X}(A)}) \\
\ge &
\sum_{a\in \Omega}
P[a,\Omega] \rE_{{\bf X}| [a,\Omega] }
\Bigl|P^A(a)+\sum_{a' (\neq a )\in f_{\bf X}^{-1}(a)} P^A(a')-\frac{1}{M}\Bigr| \\
\ge &
\sum_{a\in \Omega}
P[a,\Omega] \rE_{{\bf X}| [a,\Omega] }
\Bigl(P^A(a)+\sum_{a' (\neq a )\in f_{\bf X}^{-1}(a)} P^A(a')-\frac{1}{M}\Bigr) \\
= &
\sum_{a\in \Omega}
(1-\frac{|\Omega|}{M})
\Bigl(P^A(a) -\frac{1}{M} P^A(\Omega)\Bigr) \\
= &
(1-\frac{|\Omega|}{M})
(P^A(\Omega) -\frac{|\Omega|}{M} P^A(\Omega))
=(1-\frac{|\Omega|}{M})^2 P^A(\Omega).
\end{align*}

\section{Proof of Lemma \ref{l9-9-1}}\Label{al9-9-1}
We choose $s(R)$ such that
$
%\tilde{H}_{1+s}'(A|P)|_{s=s(R)}
\frac{d (s{H}_{1+s}(A|P))}{ds}|_{s=s(R)} 
=
H(P_{1+s(R)})+D(P_{1+s(R)}\|P) = R$,
where $P_{1+s}(a):= \frac{P(a)^{1+s}}{\sum_{a'}P(a')^{1+s}}$.
When $Q$ satisfies $H(Q)+D(Q\|P) = R$,
\begin{align*}
& D(Q\|P)-D(P_{1+s}\|P) \\
= &\sum_a Q(a) (\log Q(a)-\log P(a))\\
&-\sum_a \frac{P(a)^{1+s}}{\sum_{a'}P(a')^{1+s}}(\log \frac{P(a)^{1+s}}{\sum_{a'}P(a')^{1+s}}  -\log P(a)) \\
= & \sum_a Q(a) (\log Q(a)-\log \frac{P(a)^{1+s}}{\sum_{a'}P(a')^{1+s}})\\
&+\sum_a 
(Q(a)-\frac{P(a)^{1+s}}{\sum_{a'}P(a')^{1+s}}) \\
&\quad \cdot(\log \frac{P(a)^{1+s}}{\sum_{a'}P(a')^{1+s}}  -\log P(a)) \\
= & D(Q\|P_{1+s})
+s \sum_a 
(Q(a)-\frac{P(a)^{1+s}}{\sum_{a'}P(a')^{1+s}})\log P(a) \\
= & D(Q\|P_{1+s}) \\
&+s (H(P_{1+s})+D(P_{1+s}\|P)-H(Q)+D(Q\|P)) \\
= & D(Q\|P_{1+s})\ge 0.
\end{align*}
Hence, 
\begin{align*}
&\min_{Q:H(Q)+D(Q\|P) = R}
H(Q)+2D(Q\|P) -R \\
=&
\min_{Q:H(Q)+D(Q\|P) = R}
D(Q\|P) 
= 
 D(P_{1+s(R)}\|P)\\
=& s {H}_{1+s}(A|P)- s(R) 
\frac{d (s{H}_{1+s}(A|P))}{ds}|_{s=s(R)} 
%\tilde{H}_{1+s}'(A|P)|_{s=s(R)}
\\
=& s {H}_{1+s}(A|P)-s(R) R
= \max_{0\le s }  s {H}_{1+s}(A|P) -s R .
\end{align*}
The last equation follows from the concavity of $s {H}_{1+s}(A|P)$
concerning $s$.

Assume that $%\tilde{H}_{2}'(A|P) 
\frac{d (s{H}_{1+s}(A|P))}{ds}|_{s=1} 
\le R $.
Then, $s(R) \le 1$.
When $R' \ge R$,
\begin{align*}
&\min_{Q:H(Q)+D(Q\|P) = R'} H(Q)+2D(Q\|P) -R \\
=&\max_{0\le s }  s {H}_{1+s}(A|P) -s R +R'-R \\
\ge &
 s {H}_{1+s(R)}(A|P) -s(R) R' +R'-R \\
\ge &
 s {H}_{1+s(R)}(A|P) -s(R) R \\
=&
\max_{0\le s }  s {H}_{1+s}(A|P) -s R\\
=&
\max_{0\le s \le 1}  s {H}_{1+s}(A|P) -s R,
\end{align*}
which implies (\ref{9-9-4}).

Assume that $\frac{d (s{H}_{1+s}(A|P))}{ds}|_{s=1} 
 > R $.
When $R' \ge R$,
\begin{align*}
&\min_{Q:H(Q)+D(Q\|P) = R'} H(Q)+2D(Q\|P) -R \\
=&\max_{0\le s }  s{H}_{1+s}(A|P) -s R +R'-R \\
\ge &
 1 {H}_{1+1}(A|P) - R' +R'-R 
= 
 {H}_{2}(A|P) -R .
\end{align*}
Further, when $R'=\frac{d (s{H}_{1+s}(A|P))}{ds}|_{s=1} $,
\begin{align*}
&\min_{Q:H(Q)+D(Q\|P) = R'} H(Q)+2D(Q\|P) -R \\
=&
{H}_{1+1}(A|P) - R' +R'-R 
= 
 {H}_{2}(A|P) -R ,
\end{align*}
which implies (\ref{9-9-5}).

Further, the concavity of $s \mapsto s{H}_{1+s}(A|P)$ 
and the condition $\frac{d (s{H}_{1+s}(A|P))}{ds}|_{s=1}  > R $
imply that $\max_{0\le s \le 1}  s{H}_{1+s}(A|P) -s R = {H}_{2}(A|P) -R $.
Thus, we obtain (\ref{2-2-8}).

\section{Proof of (\ref{8-03-a})}\Label{8-03-e}
First, we consider the 
the minimum $\min_{\tilde{P}^A:H_{\min}(A|\tilde{P}^A) \ge R'} d_1(P^A,\tilde{P}^A) $,
where $\tilde{P}^A$ is chosen to be a subdistribution satisfying $H_{\min}(A|\tilde{P}^A) \ge R'$.
\begin{align}
& \min_{\tilde{P}^A:H_{\min}(A|\tilde{P}^A) \ge R'} d_1(P^A,\tilde{P}^A) \nonumber \\
=&
\min_{\tilde{P}^A:H_{\min}(A|\tilde{P}^A) \ge R'} 
\sum_{a\in {\cal A}}
|P^A(a)-\tilde{P}^A(a)| \nonumber \\
=&
\sum_{a\in {\cal A}:P^A(a)> e^{-R'} }
(P^A(a)-e^{-R}) \nonumber \\
\ge &
\frac{1}{2}
\sum_{a\in {\cal A}:P^A(a)> 2 e^{-R'} } 
P^A(a) 
= 
\frac{1}{2}
P^A \{P^A(a)> 2 e^{-R'} \}.\nonumber 
\end{align}
Using this relation,
we have
\begin{align}
& \min_{\epsilon}
(M^{\frac{1}{2}} e^{-\frac{H_{\min,\epsilon}(A|(P^{A})}{2}}
+2\epsilon) \nonumber \\
=&
\min_{R'}
\min_{\tilde{P}^A:H_{\min}(A|\tilde{P}^A) \ge R'} 
(M^{\frac{1}{2}} e^{-R'/2}
+2 d_1(P^A,\tilde{P}^A)) \nonumber \\
=&
\min_{R'}
(M^{\frac{1}{2}} e^{-R'/2}
+P^A \{P^A(a)> 2 e^{-R'} \}) \Label{8-03-b}.
\end{align}
Using (\ref{9-9-10}), we obtain
\begin{align}
& \lim_{n \to \infty}
\frac{-1}{n}\log \min_{R'}
(e^{\frac{nR-nR'}{2}} 
+(P^A)^2 \{(P^A)^n(a)> 2 e^{-nR'} \}) \nonumber \\
= &
\max_{R'}
\lim_{n \to \infty}
\frac{-1}{n}\log 
(e^{\frac{nR-nR'}{2}} 
+(P^A)^2 \{(P^A)^n(a)> 2 e^{-nR'} \})\nonumber \\
=&
\max_{R'}
\min( R'-R, \max_{0 \le s } s {H}_{1+s}(A|P^A)- s R')\nonumber \\
=&
\max_{R'}
\max_{0 \le s } 
\min( R'-R, s {H}_{1+s}(A|P^A)- s R')\nonumber \\
=&
\max_{0 \le s } 
\max_{R'}
\min( R'-R, s {H}_{1+s}(A|P^A)- s R')\Label{8-03-c}.
\end{align}
Since $s {H}_{1+s}(A|P^A)- s R'$ is monotonically decreasing with $R'$
and $R'-R$ is monotonically increasing with $R'$,
the maximum
$\max_{R'}
\min( R'-R, s {H}_{1+s}(A|P^A)- s R')$ 
is realized when $R'-R= s {H}_{1+s}(A|P^A)- s R'$,
which implies that $R'=\frac{R+2 s H_{1+s}(A|P^A)}{1+2 s}$.
Hence, 
\begin{align}
&\max_{0 \le s } 
\max_{R'}
\min( R'-R, s {H}_{1+s}(A|P^A)- s R')\nonumber \\
=&
\max_{0 \le s } \frac{s H_{1+s}(A|P^A)-s R}{1+2s}.\Label{8-03-d}
\end{align}
Therefore, combining (\ref{8-03-b}), (\ref{8-03-c}), and (\ref{8-03-d}),
we obtain (\ref{8-03-a})
because $\max_{0\le s } 
\frac{s{H}_{1+s}(A|P^{A}) -s R}{1+2s} 
=
\max_{0\le t \le 1 } 
\frac{t{H}_{1/(1-t)}(A|P^{A}) -t R}{1+t}$.

\end{document}